\documentclass[11pt]{article}

\usepackage{url}
\usepackage{hyperref}
\usepackage{amsmath}
\usepackage{amssymb}
\usepackage{graphicx}
\usepackage{float}
\usepackage{caption}
\usepackage{amsthm}
\usepackage{color}
\usepackage{tabu}
\usepackage{cleveref}
\definecolor{green}{rgb}{0,0.5977,0}
\usepackage[linesnumbered,algoruled,boxed,lined]{algorithm2e}

\newcommand{\rr}{\mathbb R}

\newcommand{\txt}[1]{\text{#1}}

\makeatletter 
\g@addto@macro{\@algocf@init}{\SetKwInOut{Parameter}{Parameters}} 
\makeatother

\theoremstyle{plain}
\newtheorem{theorem}{Theorem}
\newtheorem{lemma}[theorem]{Lemma}

\theoremstyle{definition}
\newtheorem{definition}[theorem]{Definition}

\numberwithin{theorem}{section}

\usepackage{authblk}

\title{The Point-Boundary Art Gallery Problem is $\exists\rr$-hard}

\date{April 2025}
\author[1]{Jack Stade}
\affil[1]{University of Copenhagen, Denmark}

\begin{document}
\maketitle
\begin{abstract}
    We resolve the complexity of the point-boundary variant of the art gallery problem, showing that it is $\exists\mathbb{R}$-complete, meaning that it is equivalent under polynomial time reductions to deciding whether a system of polynomial equations has a real solution. 

    The art gallery problem asks whether there is a configuration of {\it guards} that together can see every point inside of an {\it art gallery} modeled by a simple polygon. The original version of this problem (which we call the point-point variant) was shown to be $\exists\mathbb{R}$-hard [Abrahamsen, Adamaszek, and Miltzow, JACM 2021], but the complexity of the variant where guards only need to guard the walls of the art gallery was left as an open problem. We show that this variant is also $\exists\mathbb{R}$-hard.

    Our techniques can also be used to greatly simplify the proof of $\exists\mathbb{R}$-hardness of the point-point art gallery problem. The gadgets in previous work could only be constructed by using a computer to find complicated rational coordinates with specific algebraic properties. All of our gadgets can be constructed by hand and can be verified with simple geometric arguments. 
\end{abstract}

\tableofcontents

\newpage

\section{Introduction}
\subsection{Art gallery problem}
The original form of the art gallery problem (AGP) presented by Victor Klee (see O'Rourke \cite{ArtGalleryTextbook}) asks whether a simple polygonal region $P$ can be \emph{guarded} by $n$ guards. That is, whether there is a set of $n$ points (called the \emph{guards}) in $P$ such that every point in $P$ is \emph{visible} to some guard, meaning that the the line segment between that point and that guard is contained in $P$. The polygon $P$ is referred to as the \emph{art gallery}. 

The vertices of $P$ are usually restricted to rational or integer coordinates, but even so an optimal configuration might require guards with irrational coordinates. Abrahamsen, Adamaszek and Miltzow give explicit examples that require irrational coordinates in \cite{IrrationalGuards}. For this reason, we don't expect algorithms to actually output the guarding configurations, only to determine how many guards are necessary. 

It is non-trivial to see that the problem is even decidable. The first exact algorithm for the most general variants of the AGP is attributed to Sharir (see Efrat and Har-Peled \cite{FirstExact}). 

\subsection{The complexity class $\exists\rr$}

The decision problem ETR (Existential Theory of the Reals) asks whether a sentence of form:
\[\exists X_1\dots\exists X_n \Phi(X_1,\dots,X_n)\]
is true, where the $X_i$ are real variables and $\Phi$ is a formula in the $X_i$ involving $0$, $1$, $+$, $-$, $\cdot$, $=$, $<$, $\le$, $\neg$, $\wedge$, and $\vee$. The complexity class $\exists\rr$ consists of problems that can be reduced to ETR in polynomial time. A number of interesting problems have been shown to be complete for $\exists\rr$, including for example packing polygons in a square \cite{PackingProblem} and the problem of deciding whether there exists a point configuration with a given order type \cite{Mnev, Stretchability}.

By a result of Schaefer and Stefankovic \cite{ETRRobustness}, the exact inequalities used don't matter; for example we obtain an equivalent definition of $\exists\rr$ if we don't allow $=$, $\le$, or $\neg$ in $\Phi$. The same authors have more recently shown that similar results hold at every level of a \emph{real hierarchy}, of which $\exists\rr$ is the first level \cite{RealHeirarchy}.

It is straightforward to show that $\txt{NP}\subseteq\exists\rr$. It is also known, though considerably more difficult to prove, that $\exists\rr\subseteq \txt{PSPACE}$ (Canny \cite{SubsetPSPACE}). It is unknown whether either inclusion is strict. 

\subsection{Art gallery variants}

A natural class of variants of the art-gallery problem is given by restricting both the positions that can be occupied by the guards and the points that need to be guarded. We will adopt the notation used in \cite{AlmostConvex}:

\begin{definition}
(Agrawal, Knudsen, Lokshtanov, Saurabh, and Zehavi \cite{AlmostConvex}) The \emph{X-Y Art Gallery problem}, where $X, Y\in \{\text{Vertex}, \text{Point}, \text{Boundary}\}$, asks whether the polygon $P$ can be guarded with $n$ guards, where if $X=\text{Vertex}$ the guards are restricted to lie on the vertices of the polygons, if $X=\text{Boundary}$ the guards are restricted to lie on the boundary of the polygon, and if $X=\text{Point}$ then the guards can be anywhere inside the polygon. The region that must be guarded is determined by $Y$ analogously.
\end{definition}

\Cref{tab:InitialsTable} lists these variants and the known bounds on complexity.

\begin{table}[ht]
\begin{center}
\begin{tabu}{|c|c|c|}
\hline
Variant&Complexity Lower Bound&Complexity Upper Bound\\
\hline
Vertex-Y&NP\cite{NPHardness}&NP\\
X-Vertex&NP\cite{NPHardness}&NP\\
Point-Point&$\exists\rr$\cite{ExistsRHardness}&$\exists\rr$\cite{ExistsRHardness}\\
Boundary-Point&$\exists\rr$\cite{ExistsRHardness}&$\exists\rr$\cite{ExistsRHardness}\\
Point-Boundary&$\mathbf{\exists\rr}${\bf [this paper]}&$\exists\rr$\cite{ExistsRHardness}\\
Boundary-Boundary&NP\cite{NPHardness}&$\exists\rr$\cite{ExistsRHardness}\\
\hline
\end{tabu}
\caption{Variants of the art gallery problem}
\label{tab:InitialsTable}
\end{center}
\end{table}

If $X$ or $Y$ is \emph{Vertex}, then the problem is easily seen to be in NP. Lee and Lin \cite{NPHardness} showed that all of these variants are NP-hard (the result is stated for all the X-Point variants, but their constructions also work for the other variants. See also \cite{BoundaryNPHard}). More recently, Abrahamsen, Adamaszek, and Miltzow \cite{ExistsRHardness} showed that the point-point and boundary-point variants are $\exists\rr$ complete. It is straightforward to extend their proof of membership in $\exists\rr$ to any of these variants, but they list $\exists\rr$ hardness of the point-boundary variant an open problem.

\subsection{Our results}

Our main result is that the point-boundary variant of the art gallery problem is, up to polynomial time reductions, as hard as deciding whether a system of polynomial equations has a real solution.

\begin{theorem}\label{BGMain}
The point-boundary variant of the art gallery problem is $\exists\rr$-complete.
\end{theorem}

While the complexity of the boundary-boundary variant is still unsolved, this is is enough to show that the the X-Y art gallery problem is equivalent to the Y-X art gallery problem for $X, Y\in \{\text{Vertex}, \text{Point}, \text{Boundary}\}$. 

Our ideas can also be used to considerably simplify the construction in \cite{ExistsRHardness}. Both the construction and verification of each of our gadgets is simpler than that of the corresponding gadget in \cite{ExistsRHardness}. Indeed, our gadgets can be drawn by hand with little more than compass-and-straightedge-style constructions. Each gadget is specified by a simple set of geometric properties rather than by exact coordinates. Our gadgets are quite flexible and the geometry could probably be adapted to other settings.

Our reduction is from a problem called ETR-INV-REV, which is a slight modification of a problem ETR-INV from \cite{ExistsRHardness}. In \Cref{sec:ETR_INV_REV}, we define this problem and show that it is $\exists\rr$-hard. 

Given an instance $\Phi$ of ETR-INV-REV, we show how to construct a polygon whose boundary can be guarded by some number $n$ guards if and only if $\Phi$ has a satisfying assignment. In \Cref{sec:Arrangement}, we describe the overall structure of the polygon that we construct and explain how to constrain the positions of guards. In \Cref{sec:CopyNooks}, we describe copy gadgets that are required in order to use a single variable in multiple constraints. Finally, in \Cref{sec:GadgetVerification} we describe the gadgets that create constraints and prove \Cref{BGMain}.

\section{ETR-INV-REV}\label{sec:ETR_INV_REV}

The proof of \Cref{BGMain} is by reduction of the problem we call ETR-INV-REV to the point-boundary variant of the art gallery problem.

\begin{definition}
(ETR-INV-REV) In the problem ETR-INV-REV, we are given a set of real variables $\{x_1, \dots, x_n\}$ and a set of inequalities of the form:

\[x=1,\quad xy\ge 1,\quad x\left(\frac52-y\right)\le 1,\quad x+y\le z,\quad x+y\ge z,\]
for $x, y, z\in \{x_1,\dots, x_n\}$. The problem asks whether there is an assignment of the $x_i$ satisfying these inequalities with each $x_i\in[\frac12, 2]$.
\end{definition}
Abrahamsen, Adamaszek and Miltzow \cite{ExistsRHardness} proved the $\exists\rr$-hardness of the point-point and boundary-point variants using a similar problem called ETR-INV.

\begin{definition}(Abrahamsen, Adamaszek and Miltzow \cite{ExistsRHardness})
(ETR-INV) In the problem ETR-INV, we are given a set of real variables $\{x_1, \dots, x_n\}$ and a set of equations of the form:

\[x=1,\quad xy=1,\quad x+y=z,\]
for $x, y, z\in \{x_1,\dots, x_n\}$. The problem asks whether there is a solution to the system of equations with each $x_i\in[\frac12, 2]$.
\end{definition}

\begin{theorem}(Abrahamsen, Adamaszek, Miltzow \cite{ExistsRHardness})
The problem ETR-INV is $\exists\rr$-complete.
\end{theorem}

Regardless of the specific art gallery variant, it seems to be difficult to make a construction that admits both $xy\le 1$ and $xy\ge 1$ constraints. The $xy\le 1$ inversion gadget in \cite{ExistsRHardness} is actually closer to a $x\left(\frac52-y\right)\le 1$ gadget, and another gadget is used to compute a constraint like $y'=\left(\frac52-y\right)$. By reducing from ETR-INV-REV instead, our construction avoids the need for a reversing gadget.

\begin{theorem}\label{ETRRevHardness}
The problem ETR-INV-REV is $\exists\rr$-hard.
\end{theorem}

\begin{proof}
For a variable $x$, we will show how to construct some additional variables, including a variable $V$, as well as some constraints that can only be satisfied if $V=\frac52-x$. These constraints should also be satisfiable for any value of $x\in I$ for some interval $I$ of length at most $1$ (the constraints depend on the choice of interval $I$).

Miltzow and Schmiermann \cite{ConvexConcave} show that the addition and constant constraints are sufficient to construct a variable equal to any rational number in $[\frac12, 2]$. We will describe the construction for $I=[1, 2]$, and the constraints for any other interval can be obtained by translating.

Using the construction from \cite{ConvexConcave}, we construct variable with values $\frac32$ and $\frac52$. We also add variables $V_1$, $V_2$, $V_3$ and $V$ and the following constraints:

\[V_1+V_1=x\quad\left(V_1=\frac12x\in[\frac12, 1]\right)\]
\[V_1+V_2=\frac32\quad\left(V_2=\frac32-\frac12x\in[\frac12, 1]\right)\]
\[V_3=V_2+V_2\quad\left(V_3=3-x\in[1, 2]\right)\]
\[V+\frac12=V_3\quad\left(V=\frac52-x\in[\frac12, \frac32]\right)\]

Examining the proof of $\exists\rr$-hardness of ETR-INV in \cite{ExistsRHardness}, we can see that it produces instances of ETR-INV with the following property: every time a $xy=1$ constraint appears, there is an interval $I$ of length at most $1$ such that $x\in I$ in any satisfying assignment of the instance. The interval $I$ is always one of $[\frac{8}{15}, \frac{8}{13}], [\frac{8}{7}, \frac{8}{5}]$ or $[\frac{65}{64}, \frac{105}{64}]$ (and is known when each $xy=1$ constraint is constructed). 

Starting with such an ETR-INV instance $\Phi$, we construct an ETR-INV-REV instance $\Psi$. For each $x+y=z$ constraint in $\Phi$, we put constraints $x+y\le z$ and $x+y\ge z$ in $\Psi$. For each $xy=1$ constraint, we add constraints $xy\ge 1$ and $y\left(\frac52-V\right)\le 1$ for $V$ constructed as above.

$\Psi$ has a satisfying assignment if and only if $\Phi$ does. This reduction can be performed in polynomial time, so ETR-INV-REV is $\exists\rr$-hard.
\end{proof}

\section{Arranging the art gallery}\label{sec:Arrangement}

Throughout this section, and the rest of the paper, $AB$ refers to a line segment with endpoints $A$ and $B$, $\overleftrightarrow{AB}$ is the line containing that segment, and $|AB|$ is the length of that segment.

\subsection{Creating guard regions}

We will designate some number $n$ of disjoint \emph{guard regions} inside the art gallery such that any guarding configuration with $n$ guards must have exactly one guard in each region. Each guard region is determined by \emph{wedges} on the polygon boundary. Each wedge is formed by two edges of the art gallery meeting at a convex corner. The \emph{visibility region} of a wedge is the set of points that can see the tip of the wedge. A wedge and its visibility region is shown in \Cref{fig:RegionSegment} (left).

\begin{figure}
\centering
\includegraphics[page=1]{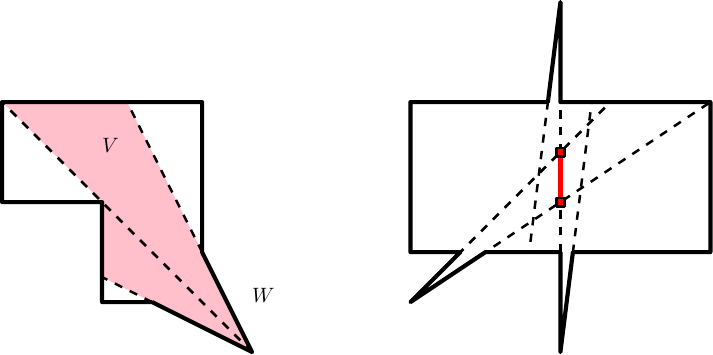}
\caption{Left: A wedge ($W$) and its visibility region ($V$). Right: The intersection of the visibility regions of the three wedges shown forms a guard segment. This method of using $3$ wedges to form a guard segment is due to Bertschinger, El Maalouly, Miltzow, Schnider and Weber \cite{HomotopyUniversality} and is a slight simplification of the method for forming guard segments from \cite{ExistsRHardness}.}
\label{fig:RegionSegment}
\end{figure}

In a guarding configuration, the visibility region of each wedge must contain a guard. We designate a guard region by specifying some (non-zero) number of wedges, and the guard region is the intersection of the visibility regions of those wedges. These must be chosen that visibility regions of any two wedges \emph{corresponding to different guard regions} do not intersect. A guard region shaped like a line segment is called a \emph{guard segment}, as in \Cref{fig:RegionSegment} (right).

\begin{lemma}\label{lem:GuardRegions}
If we designate $n$ guard regions, then any guarding configuration has at least $n$ guards, and a guarding configuration with exactly $n$ guards has $1$ guard in each guard region.
\end{lemma}

\begin{proof}
Choose wedges $w_1,\dots,w_n$, one from each of the guard regions. The guard regions are chosen so that the visibility regions of these wedges do not intersect. In a guarding configuration, each of these regions contains at least $1$ guard. Since these region do not intersect, this requires at least $n$ guards.

If a guarding configuration contains exactly $n$ guards, then there will be exactly $1$ guard in the visibility region of $w_i$ for each $i$. Consider the guard region containing $w_j$ and suppose $v$ is another wedge forming that guard region. There is a guard in the visibility region of $v$, and this region doesn't intersect the visibility region of $w_i$ for $i\ne j$, so this guard must be the one in $w_j$. Repeating this argument for the other wedges forming this guard region, we see that the guard in $w_j$ must be in each of them. So the guard is in the guard region.

So a guarding configuration with exactly $n$ guards has $1$ guard in each guard region.
\end{proof}

The guard segments in our construction will all be vertical, and the position of a guard on a guard segment will encode a variable in the interval $[\frac12, 2]$, with the bottom endpoint of the interval corresponding to $\frac12$ and the top endpoint corresponding to $2$. 

\subsection{Constraint nooks}

We use constraint \emph{nooks} to enforce constraints involving multiple variables. Each nook consists of a line segment on the boundary of the art gallery, called the \emph{nook segment}, and two polygon vertices called \emph{opening points}. The \emph{region of partial visibility} of a nook is the set of points which can see some part of the nook segment.

If the nook segment has endpoints $A$ and $B$ and the opening points are $E$ and $F$ where $E, A, B, F$ occur in that order on the polygon boundary, then we say that $E$ \emph{restricts} $A$ and $F$ \emph{restricts} $B$ (see \Cref{fig:NookCollab}).

These are slightly more general than the nooks from \cite{ExistsRHardness}, in that it will sometimes be necessary to intersect a wedge with a nook, as in \Cref{fig:NookCollab} (left).

\begin{figure}
\centering
\includegraphics[page=2]{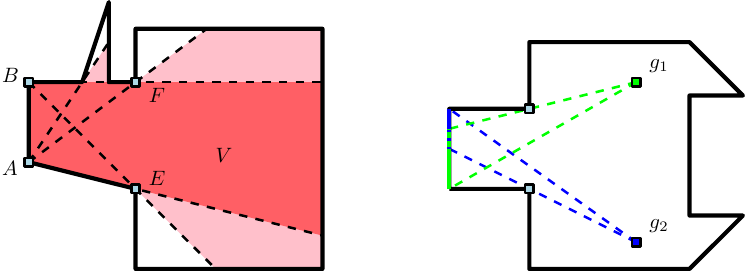}
\caption{Left: A nook with nook segment $AB$, opening points $E$ restricting $A$ and $F$ restricting $B$, and partial visibility region $V$. The outer parts of $V$ (lighter) can only see part of the nook segment. Right: This art gallery can be guarded by two guards, but there is no such configuration where the nook segment on the left is guarded entirely by a single guard. This creates a continuous constraint between the positions of the two guards.}
\label{fig:NookCollab}
\end{figure}

A guarding configuration will have some non-zero number of guards in the partial visibility region of a nook, that together must guard the nook segment, as in \Cref{fig:NookCollab} (right). Multiple guards can collaborate to guard a single nook segment, creating a continuous dependence between the positions of those guards.

\subsection{Schematic of the art gallery}

\begin{figure}
\centering
\includegraphics[page=3]{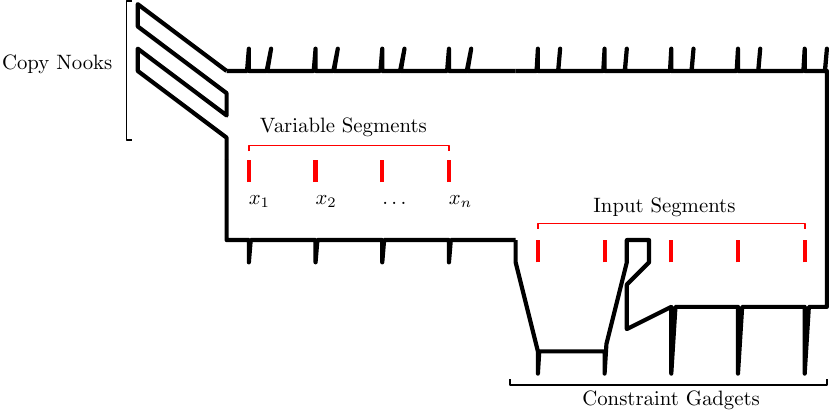}
\caption{A schematic of the art gallery that we construct (not to scale). The variables are encoded by positions of guards on the variable segments (red, middle). The constraint gadgets on the bottom right enforce constraints like $xy\ge 1$ or $x+y\le z$ on the variables encoded by the input segments (red, bottom right). The copy nooks copy the values of the appropriate variables onto the input segments. The wedges creating the segment-shaped guard regions can be seen along the top and bottom walls of the art gallery.}
\label{fig:FullSchematic}
\end{figure}

\Cref{fig:FullSchematic} shows a schematic of the entire art gallery construction. We start with an instance $\Phi$ of ETR-INV-REV with variables $x_1, \dots, x_n$ and constraints of form:

\[x=1,\quad xy\ge 1,\quad x\left(\frac52-y\right)\le 1,\quad x+y\le z,\quad x+y\ge z,\]

Each variable is represented by a guard region called a \emph{variable segment}. The variable segments are arranged in a row near the middle of the art gallery. The height of a guard on a variable segment corresponds linearly to the value of the variable encoded, with the bottom of the variable segment mapping to $\frac12$ and the top mapping to $2$.

The constraints of form $x=1$ in $\Phi$ are created by adding wedges to create a guard region consisting of only a single point, as shown in \Cref{fig:GuardPoint}. Each constraint involving more than one variable is created by a \emph{constraint gadget}, which we place in a row on the bottom right of the art gallery. The constraint gadgets are modular in the sense that they (almost) do not depend on $\Phi$. Each constraint gadget operates on either $2$ or $3$ guard segments called \emph{input segments}. Nooks in the constraint gadget create constraints on the variables represented by the guards on the input segments. 

Each input segment can only appear in one constraint gadget, so we need \emph{copy nooks} to relate each input segment to the variable segment for the variable that it represents. 

\begin{figure}
\centering
\includegraphics[page=4]{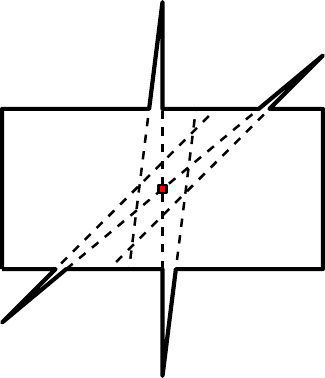}
\caption{Four wedges create a point-shaped guard region. }
\label{fig:GuardPoint}
\end{figure}

\section{Construction of the copy nooks}\label{sec:CopyNooks}

If a guard on a variable segment encodes a value $x$ and a guard on an input segment encodes a value $y$, then a copy nook creates a constraint like $x\le y$ or $x\ge y$. \Cref{fig:CopyNook} shows an example of a copy nook enforcing a constraint like $x\ge y$. This type of copy nook is due to Stade and Tucker-Foltz \cite{TopologyUniversality}.

\subsection{Verifying a single copy nook}

\begin{figure}
\centering
\includegraphics[page=5]{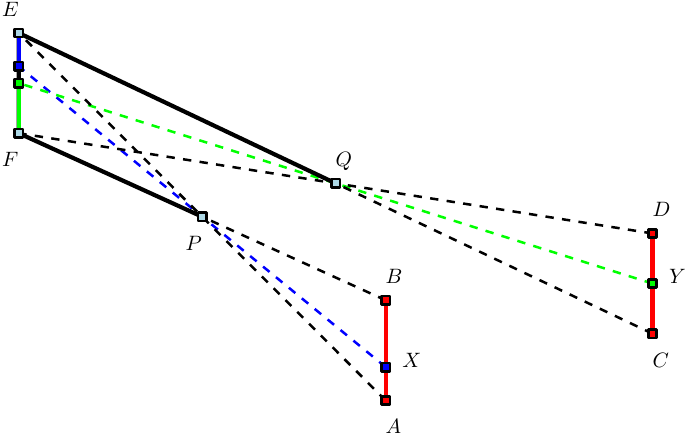}
\caption{A copy nook, shown in the situation where the guard positions don't fully guard the nook. }
\label{fig:CopyNook}
\end{figure}

\begin{lemma}\label{lem:CopyLemma}
Suppose segments $AB$ and $CD$ are such that $\overleftrightarrow{AB}$ and $\overleftrightarrow{CD}$ are parallel, and suppose $\overleftrightarrow{AC}$ and $\overleftrightarrow{BD}$ intersect at a point $P$, as in \Cref{fig:ParallelCopy}. If a line through $P$ intersects $AB$ at a point $X$ and intersects $CD$ at a point $Y$, then $\frac{|AX|}{|AB|}=\frac{|CY|}{|CD|}$.
\end{lemma}

\begin{proof}
Triangles $APB$ and $CPD$ are similar, so $\frac{|AP|}{|AB|}=\frac{|CP|}{|CD|}$. Also, the triangles $AXP$ and $CYP$ are similar, so $\frac{|AX|}{|AP|}=\frac{|CY|}{|CP|}$. Multiplying, we obtain $\frac{|AX|}{|AB|}=\frac{|CY|}{|CD|}$.
\end{proof}

\begin{figure}
\centering
\includegraphics[page=6]{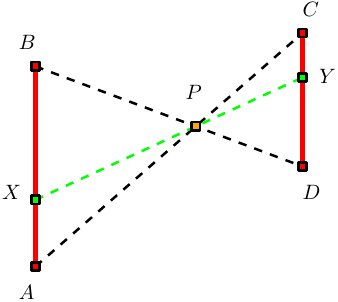}
\caption{Setup for \Cref{lem:CopyLemma}. }
\label{fig:ParallelCopy}
\end{figure}

\begin{lemma}\label{lem:CopyNook}
Suppose that $AB$ and $CD$ are parallel guard segments and there is a nook with nook segment $EF$ that is parallel to $AB$ and $CD$. Suppose the nook has openings points $P$ on the intersection point of $\overleftrightarrow{AE}$ and $\overleftrightarrow{BF}$ and $Q$ on the intersection of $\overleftrightarrow{BE}$ and $\overleftrightarrow{CF}$ so that $P$ restricts $F$ and $Q$ restricts $G$.

Suppose also that the nook is unobstructed, that is that every segment through $P$ from a point on $AB$ to a point on $EF$ is contained in the art gallery and segments through $Q$ are similarly unobstructed.

Then guards at positions $X$ and $Y$ on $AB$ and $CD$ respectively guard the nook if and only if:

\[\frac{|AX|}{|AB|}\ge \frac{|CY|}{|CD|}\]
\end{lemma}

The setup for \Cref{lem:CopyNook} is shown in \Cref{fig:CopyNook}. The proof is a straightforward application of \Cref{lem:CopyLemma}.

\begin{figure}
\centering
\includegraphics[page=7]{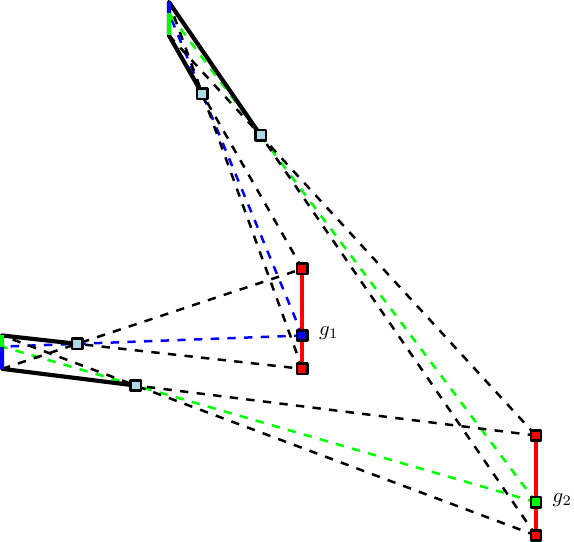}
\caption{A guard ($g_1$) on a variable segment and a guard ($g_2$) on an input segment are related by two copy nooks. The upper nook is guarded when the value encoded by $g_1$ is at least the value encoded by $g_2$, and the lower nook is guarded when the value encoded by $g_1$ is at most the value encoded by $g_2$.}
\label{fig:BothCopyNooks}
\end{figure}

By changing which intersection point restricts which endpoint of the nook segment, we can create a constraint either $x\le y$ or $x\ge y$. \Cref{fig:BothCopyNooks} shows a pair of guard segments with both types of copy nook between them.

Enforcing a constraint $x=y$ requires $2$ nooks, but all of the constraints allowed by ETR-INV-REV are monotone with respect to any single variable, so we only need a single copy nook for each input segment. For example, $xy\ge 1$ is equivalent to $x'y'\ge 1$, $x\ge x'$ and $y\ge y'$ (when there are no other constraints on $x'$ and $y'$).

\subsection{Arranging all the copy nooks}

Here we show that it is possible to arrange all the copy nooks in such a way that none of them interfere with each other. We construct $n$ variable segments and $i$ spaces for input segments. All the guard segments are vertical and have length $1$, and both sets of segments occur at regular intervals with spacing $d$ in two horizontally aligned rows, which are separated by a vertical distance $h$. Suitable values of $h$ and $d$ will be determined later. Let $A_k$ be the bottom of the $k$th variable segment and let $B_k$ be top of this segment, and let $C_jD_j$ be the segment at the $j$th position in the row of input segments, as in \Cref{fig:SegmentBanks}.

\begin{figure}
\centering
\includegraphics[page=8]{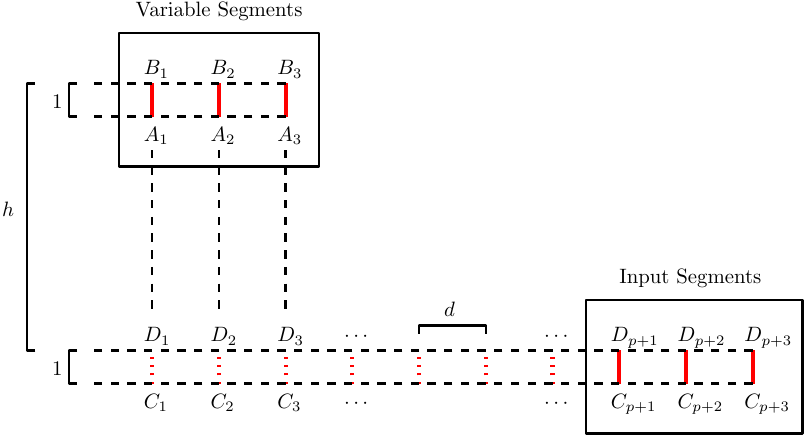}
\caption{Specification of the input segments and variable segments. The possible places for input segments occur at regular intervals, but the constraint gadgets might not allow an input segment to be put in each space. The constraint gadgets don't depend on the copy nooks at all, so we know how many valid spaces for input segments are needed at the time that the copy nooks are being constructed. }
\label{fig:SegmentBanks}
\end{figure}

Between a given variable segment and input segment, there are two types of nook which we might need to construct. If the guard on the variable segment represents a value $y$ and the guard on the input segment represents a variable $x$, then a $\ge$ copy nook (lower nook in \Cref{fig:BothCopyNooks}) enforces the constraint $x\ge y$, and a $\le$ copy nook (upper nook in \Cref{fig:BothCopyNooks}) enforces the constraint $x\le y$. 

\begin{figure}
\centering
\includegraphics[page=9]{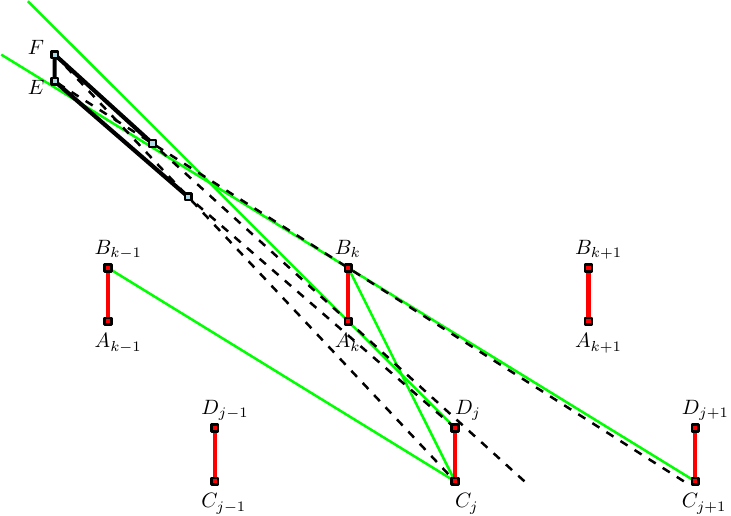}
\caption{A $\ge$ copy nook that restricts a guard on $C_jD_j$ to be above the relative position of the guard on $A_kB_k$. If the nook segment is above the line $\overleftrightarrow{C_{j+1}B_k}$ and below the line $D_jA_k$ as shown, then its partial visibility region doesn't intersect any segments that it isn't supposed to.}
\label{fig:GENookBound}
\end{figure}

\begin{figure}
\centering
\includegraphics[page=10]{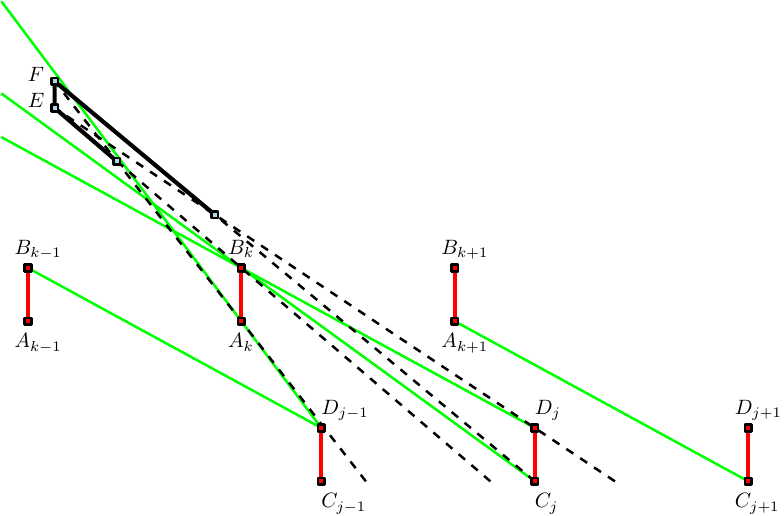}
\caption{A $\le$ copy nook that restricts the guard on $C_jD_j$ to be below the relative position of the guard on $A_kB_k$. The nook segment should be above $\overleftrightarrow{C_jB_k}$ and below $\overleftrightarrow{D_{j-1}A_k}$.}
\label{fig:LENookBound}
\end{figure}

\Cref{lem:SingleNook} is the first step in showing the all the appropriate nooks can be placed.

\begin{lemma}\label{lem:SingleNook}
Suppose for some $k$ and $j$ we have that $\overleftrightarrow{B_kC_{j+1}}$ doesn't intersect $A_{k+1}B_{k+1}$. Then if the nook segment $EF$ for a $\ge$ copy nook between variable segment $A_kB_k$ and input segment $C_jD_j$ is placed above the line $\overleftrightarrow{B_kC_{j+1}}$ and below the line $\overleftrightarrow{A_kD_j}$ then no obstructions occur. That is, $\overleftrightarrow{EB_k}$ doesn't intersect $A_{k+1}B_{k+1}$ or $C_{j+1}D_{j+1}$, $\overleftrightarrow{FC_j}$ doesn't intersect $A_{k-1}B_{k-1}$ or $C_{j-1}D_{j-1}$, and $\overleftrightarrow{FA_k}$ doesn't intersect $ED_j$. Furthermore, any line between the nook segment and either guard segment is steeper than $\overleftrightarrow{B_kC_{j+1}}$ and less steep than $\overleftrightarrow{B_kC_j}$.

Similarly, if $\overleftrightarrow{B_kD_j}$ doesn't intersect $A_{k+1}B_{k+1}$ then a $\le$ copy nook between variable segment $A_kB_k$ and input segment $C_jD_j$ will have no obstructions if the nook segment $EF$ is above the line $\overleftrightarrow{B_kC_j}$ and below the line $\overleftrightarrow{A_kD_{j-1}}$, and the lines between the nook segment and its guard segments have slope bounded by $\overleftrightarrow{B_kC_j}$ and $\overleftrightarrow{A_kD_{j-1}}$.
\end{lemma}

\begin{proof}
See \Cref{fig:GENookBound,fig:LENookBound}. For the $\ge$ nook, we know that $\overleftrightarrow{FC_j}$ doesn't intersect $A_{k-1}B_{k-1}$ since $\overleftrightarrow{FC_j}$ is steeper than $\overleftrightarrow{B_kC_{j+1}}$ and the lines $\overleftrightarrow{B_kC_{j+1}}$ and $\overleftrightarrow{B_{k-1}C_j}$ are parallel. The $\le$ case is similar.
\end{proof}

In order to prevent the copy gadgets from interfering with the constraint gadgets later, it will be helpful to have a bound on the slopes of the lines bounding the partial visibility regions of each copy nook. We will construct these nooks so that the partial visibility regions are bounding by lines with slope between $-\frac12$ and $-1$. By \Cref{lem:SingleNook}, this holds for a $\le$ or $\ge$ nook between segments $A_kB_k$ and $C_jD_j$ as long as lines $\overleftrightarrow{B_kC_j}$, $\overleftrightarrow{B_kC_{j+1}}$, $\overleftrightarrow{A_kD_{j-1}}$, and $\overleftrightarrow{A_kD_j}$ have slopes in this range.

For general $k$ and $j$, the line $\overleftrightarrow{B_kC_j}$ has slope $-\frac{h+1}{d(j-k)}$ and the line $\overleftrightarrow{A_kD_j}$ has slope $-\frac{h-1}{d(j-k)}$. These are between $-\frac12$ and $-1$ so long as:

\begin{equation}
\label{HBound}
\frac12d(j-k)+1\le h\le d(j-k)-1
\end{equation}

There $n$ variable segments and $i$ input segments, giving $n+i-1$ values of $j-k$. However, we also need these bounds to hold for one value of $j$ beyond the row of input segments in either direction, so we should choose $h$ so that the largest and smallest values of $j-k$ which satisfy (\ref{HBound}) differ by at least $n+i+1$. These values of $j-k$ are:
\begin{equation}
\label{JKRange}
\frac{h+1}{d}\le j-k\le 2\frac{h-1}{d}
\end{equation}
So as long as $d\ge 1$, it is sufficient to have $h\ge (n+i+2)d$. 

If $d>2$, these conditions on the slopes of $\overleftrightarrow{B_kC_j}$ and $\overleftrightarrow{A_kD_j}$ also ensure that $\overleftrightarrow{B_kC_{j+1}}$ and $\overleftrightarrow{B_kD_j}$ don't intersect $A_{k+1}B_{k+1}$ for these values of $j-k$.

This tells us how to construct nooks each of which only sees the correct pair of segments, but not how to construct all the nooks without them interfering with each other. The partial visibility regions of multiple nooks can intersect without issue, but the nook segment and the two walls of a nook must not occlude the partial visibility region of any other nook. \Cref{lem:NookWall} and \Cref{lem:Tube} will help us complete the construction.

\begin{lemma}\label{lem:NookWall}
If a $\le$ or $\ge$ copy nook between $A_kB_k$ and $C_jD_j$ has nook segment $EF$ with length $\alpha$ and the horizontal distance between the nook segment and $\overleftrightarrow{C_jD_j}$ is $\beta$, then all the walls of the nook are a horizontal distance at least $\beta\frac{1}{1+\alpha}$ from $\overleftrightarrow{C_jD_j}$, as in \Cref{fig:NookWall}. In particular this distance can be made arbitrarily large by increasing $\beta$.
\end{lemma}

\begin{proof}
Calculate the intersection point of $\overleftrightarrow{ED_j}$ and $\overleftrightarrow{FC_j}$ (see \Cref{fig:NookWall}).
\end{proof}

\begin{figure}
\centering
\includegraphics[page=11]{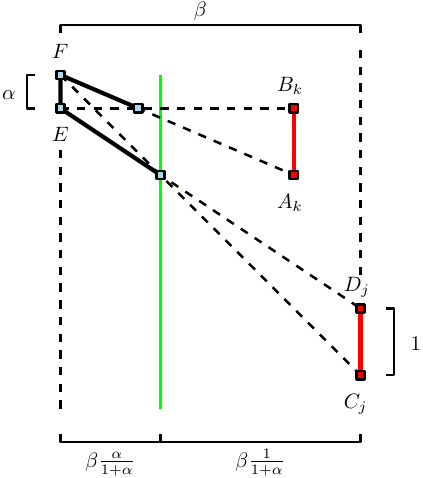}
\caption{Illustration of \Cref{lem:NookWall}.}
\label{fig:NookWall}
\end{figure}

\begin{lemma}\label{lem:Tube}
Suppose lines $\ell_1$ starting at $B_k$ and $\ell_2$ starting at $D_j$ are parallel, as in \Cref{fig:Tube}. Then if a $\le$ or $\ge$ nook between segments $A_kB_k$ and $C_jD_j$ has a nook segment which is in the tubular region bounded by $\ell_1$ and $\ell_2$, then all points in the partial visibility region of the nook which are to the left of $\overleftrightarrow{A_kB_k}$ are inside the same tubular region.
\end{lemma}

\begin{proof}
See \Cref{fig:Tube}.
\end{proof}

\begin{figure}
\centering
\includegraphics[page=12]{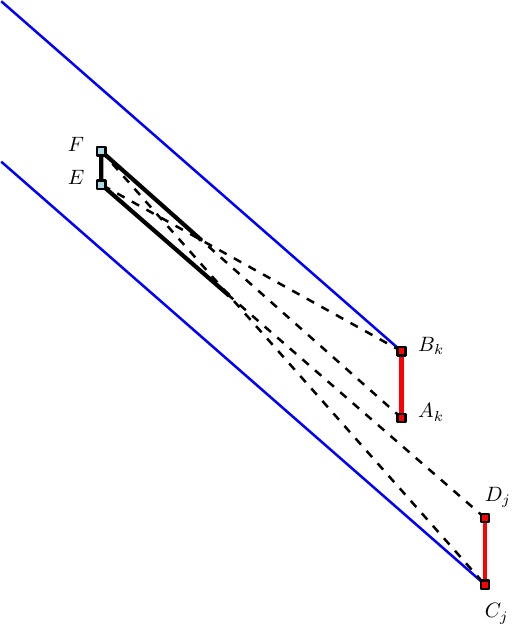}
\caption{Illustration of \Cref{lem:Tube}.}
\label{fig:Tube}
\end{figure}

We can now construct all the copy nooks.

\begin{lemma}\label{lem:TotalCopy}
Suppose there is some $p$ such that for any $1\le k\le n$ and $p\le j\le p+i$ we have that $\overleftrightarrow{A_kD_j}$ is steeper than $\overleftrightarrow{B_kC_{j+1}}$ by a slope strictly greater than some $\epsilon$ (independent of $k$ and $j$), and these lines meet the conditions from \Cref{lem:SingleNook}. Then we can construct any number of non-interfering $\ge$ or $\le$ copy nooks between variable segments $A_kB_k$ and input segments $C_jD_j$ for $1\le k\le n$ and $p+1\le j\le p+i$.
\end{lemma}

\begin{proof}
For a copy nook between $A_kB_k$ and $C_jD_j$, we want to create a tube in order to apply \Cref{lem:Tube}. We should to be able to place a nook segment arbitrarily far along the tube without violating the assumptions of \Cref{lem:SingleNook}. For a $\ge$ nook, the lines $\ell_1$ and $\ell_2$ should have slope steeper than $\overleftrightarrow{A_kD_j}$ and less steep than $\overleftrightarrow{B_kC_{j+1}}$, as in \Cref{fig:TubeAndRegion}. For a $\le$ nook, we instead bound by lines $\overleftrightarrow{A_kD_{j-1}}$ and $\overleftrightarrow{B_kC_{j}}$. By the conditions of the slopes of $\overleftrightarrow{A_kD_j}$ and $\overleftrightarrow{B_kC_{j+1}}$, we can always choose this tube so that the slope is an integer multiple of $\frac12\epsilon$. 

\begin{figure}
\centering
\includegraphics[page=13]{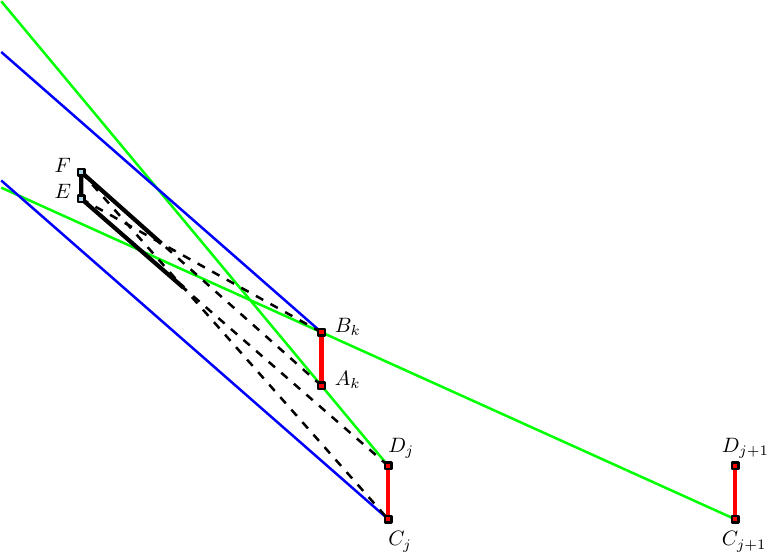}
\caption{A tube for \cref{lem:Tube} which is contained in the region from \Cref{lem:SingleNook} for points sufficiently far to the left.}
\label{fig:TubeAndRegion}
\end{figure}

Since $\overleftrightarrow{A_kD_j}$ is steeper than $\overleftrightarrow{A_kD_{j+1}}$, the only way for two parallel tubes chosen this way to intersect is if there is a $\ge$ copy nook between $A_kB_k$ and $C_jD_j$ and a $\le$ copy nook between $A_kB_k$ and $C_{j+1}D_{j+1}$. If this happens, then the regions for placing these nook segments coincide, but we choose the two tubes for these nooks to have different slopes.

\begin{figure}
\centering
\includegraphics[page=14]{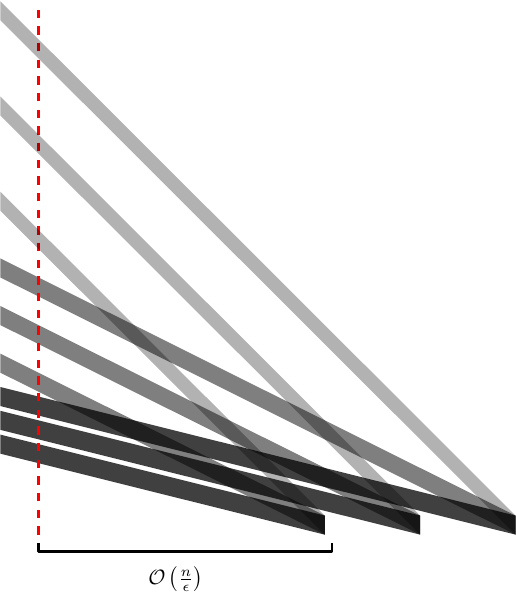}
\caption{With the tubes chosen, we can place the copy nooks so that the walls of each nook are left of the red line.}
\label{fig:AwayTubes}
\end{figure}

We can now choose a vertical line $v$ such that no point left of $v$ is in more than one tube, as in \Cref{fig:AwayTubes}. The horizontal distance from $A_1B_1$ to $v$ is at most $\mathcal{O}(n\epsilon^{-1})$. Using \Cref{lem:NookWall}, we can construct a nook in each tube such that the walls of the nook are on the left of $v$, so by \Cref{lem:Tube}, none of these nooks interfere with each other, as required.

\end{proof}

What remains to show is that for any $n$ and $i$ we can choose $h$ and $d$ so that there exists a $p$ which satisfies the conditions of \Cref{lem:TotalCopy}. The difference in slopes between $\overleftrightarrow{A_kD_j}$ and $\overleftrightarrow{B_kC_{j+1}}$ is:

\[\frac{h-1}{d(j-k)}-\frac{h+1}{d(j+1-k)}=\frac{h-1-2(j-k)}{d(j-k)(j+1-k)}\]

So we need $h>2(j-k)+2$, so that this is larger than $\epsilon=d(h(h+d))^{-1}\le (d(j-k)(j+1-k))^{-1}$ (by (\ref{JKRange})). Since $h\ge (n+i+2)d$, $j-k$ will always be at least $3$ by (\ref{JKRange}). So if we choose $d>4$, then $\frac12d(j-k)>2(j-k)+2$, meaning $h>2(j-k)+2$ is satisfied by any value of $h$ satisfying (\ref{HBound}). So we can choose $h\ge (n+i+2)d$ and $d>4$. Since the value of $d$ doesn't need to depend on $n$ or $i$, we will fix $d=5$.

\section{The constraint gadgets}\label{sec:GadgetVerification}

We start with some specifications that a constraint gadget should adhere to in order to be compatible with the copy nooks. There are several bad cases that we need to prevent:

\begin{itemize}
    \item The constraint gadget obstructs the visibility between an input segment and a copy nook.
    \item A guard on a variable segment can see part of a nook segment in the constraint gadget.
    \item An guard used by the gadget (called an \emph{auxiliary guard}) can see part of the nook segment of a copy nook.
\end{itemize}

For each constraint gadget, we designate $R_1$ and $R_2$, shown in \Cref{fig:ConstraintGadget}. The region $R_1$ is bounded by the convex hull of the input segments to that gadget and lines of slope $-\frac12$ and $-1$ through the bottom of the leftmost input segment and the top of the rightmost input segment respectively. $R_2$ is the set of points that can be connected to a point outside of the constraint gadget by a line segment that is contained in the polygon and has slope between $-1$ and $0$.

\begin{figure}
\centering
\includegraphics[page=15]{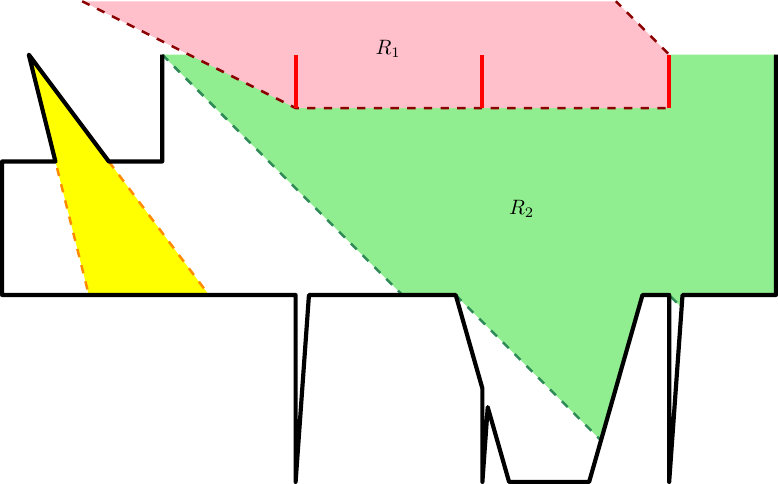}
\caption{Diagram of a constraint gadget. The region $R_1$ is bounded by lines with slope $-1$ and $-\frac12$. The region $R_2$ is bounded by lines with slope $-1$ and the polygon wall. The yellow region is a an auxiliary guard region, which must not intersect $R_2$.}
\label{fig:ConstraintGadget}
\end{figure}

\begin{lemma}\label{lem:Region1}
The region $R_1$ contains all possible sight lines between an input segment and the nook segment of one of the copy nooks, so if the constraint gadget does not have any walls that intersect this region then it will not obstruct the copy gadgets. 
\end{lemma}

\begin{proof}
The construction in \Cref{sec:CopyNooks} ensures that all sight lines between an input segment and a nook segment have slope between $-1$ and $-\frac12$, and so are contained in the region $R_1$.
\end{proof}

\begin{lemma}\label{lem:Region2}
The only points in a constraint gadget that can be seen by a variable segment are in $R_2$. So any nook with a nook segment that doesn't intersect $R_2$ must be guarded by guards on input segments or auxiliary guards in that gadget.

Points in a constraint gadget outside of $R_2$ cannot see any part of the nook segment of any copy nook. So if any auxiliary guard regions don't overlap $R_2$, then a guard in that region can't help guard any of the copy nooks. 
\end{lemma}

\begin{proof}
The construction of the copy nooks ensures that any line between a variable segment and a point at the top of a constraint gadget has slope between $-1$ and $0$. Any sight line between a point on a variable segment and a point in the constraint gadget must pass through the top, so any such line has slope between $-1$ and $0$. So $R_2$ contains all points in the constraint gadget that can be seen by a point on a variable segment.

The construction of the copy nooks also ensures that the partial visibility region of each copy nook is bounded by lines with slope between $-1$ and $-\frac12$. So the only points in a constraint gadget that can see part of a nook segment of a copy nook are in $R_2$.
\end{proof}

\subsection{Inversion gadgets}\label{sec:InversionVerification}

Each inversion gadget consists of a single constraint nook interacting with the two input variables. \Cref{fig:InversionGE,fig:InversionLE} show the $xy\ge 1$ and $x\left(\frac{5}2-y\right)\le 1$ inversion gadgets respectively. \Cref{fig:InversionGESkeleton,fig:InversionLESkeleton} show the geometry of these gadgets in more detail.

\begin{figure}
\centering
\includegraphics[page=16]{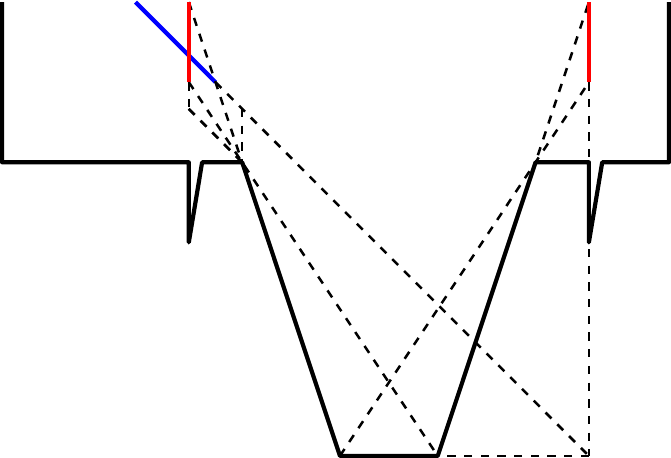}
\caption{The $xy\ge 1$ inversion gadget. The two red segments are the input segments, with the left segment representing the variable $x$ and the right segment representing $y$. Both guards see more of the nook segment the higher they are. The blue segment is a ``phantom'' segment that is used to help verify the gadget.}
\label{fig:InversionGE}
\end{figure}

\begin{figure}
\centering
\includegraphics[page=17]{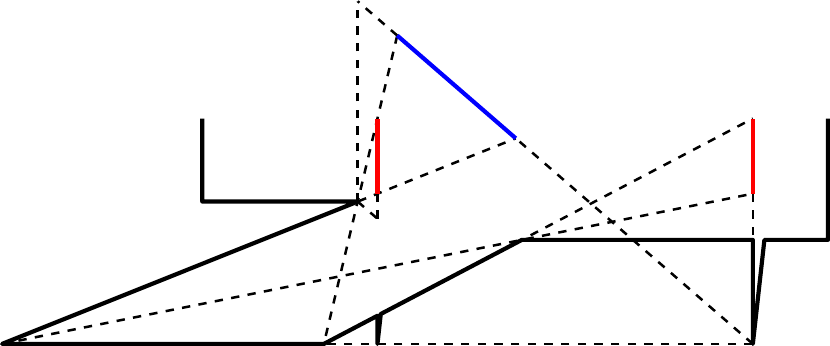}
\caption{The $x\left(\frac52-y\right)\le 1$ inversion gadget. Compared to the $xy\ge 1$ gadget, the guard on the left segment (representing $x$) now sees more of the nook segment when it is lower.}
\label{fig:InversionLE}
\end{figure}

In either gadget, we suppose that there is a guard at $X$ on the input segment $GH$. The nook has opening points $P$ and $Q$, so the guard at $X$ can see all the points to the left of point $Y$ on the nook segment $IJ$. For the $xy\ge 1$ gadget, this means that a guard on $AB$ see the rest of $IJ$ if and only if it is above the point $W$. For the $x\left(\frac{5}2-y\right)\le 1$ gadget, a guard on $AB$ should be below $W$ instead.

\begin{figure}
\centering
\includegraphics[page=18]{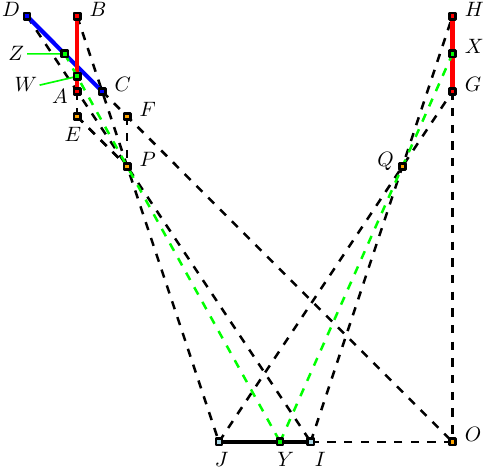}
\caption{Labeled diagram of the $xy\ge 1$ inversion gadget shown in \Cref{fig:InversionLE}.}
\label{fig:InversionGESkeleton}
\end{figure}

\begin{figure}
\centering
\includegraphics[page=19]{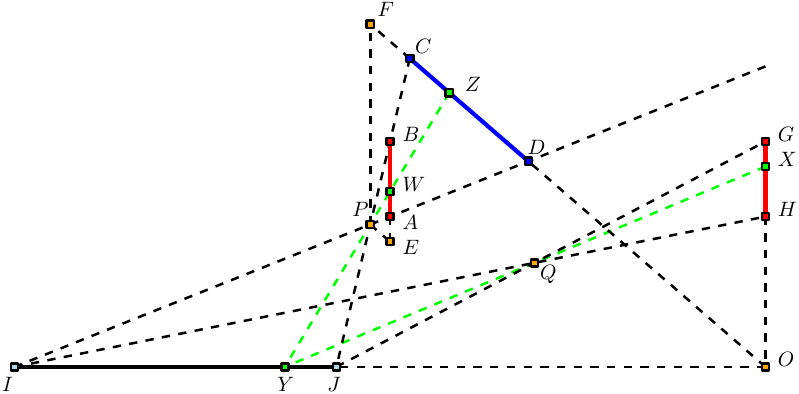}
\caption{Labeled diagram of the $x\left(\frac{5}2-y\right)\le 1$ inversion gadget shown in \Cref{fig:InversionGE}.}
\label{fig:InversionLESkeleton}
\end{figure}

In both cases, the guard on $AB$ represents the input $x$ and the guard on $GH$ represents the input $y$. The segments represent the intervals $[\frac12, 2]$, with lower points on the segments corresponding to lower values. So we should show that:

\[\left(\frac32\cdot\frac{|AW|}{|AB|}+\frac12\right)\left(\frac32\cdot\frac{|GX|}{|GH|}+\frac12\right)=1\]
for any $X$ on $GH$ and $W$ on $AB$ as shown. 

\begin{figure}
\centering
\includegraphics[page=20]{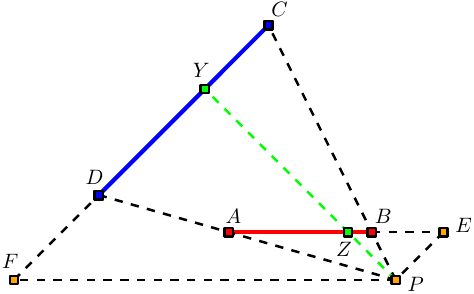}
\caption{As long as the points $P$, $Y$, and $Z$ are colinear, then the value of $|EZ||FY|$ does not depend on the positions of $Y$ and $Z$, so the positions of $Z$ and $Y$ are related by inversion. The curved relationship occurs because $AB$ and $CD$ are not parallel.}
\label{fig:InversionLemma1}
\end{figure}

\begin{lemma}\label{lem:Inversion}
If line segments $AB$ and $CD$ are not parallel, as in \Cref{fig:InversionLemma1}, let $P$ be the intersection of $\overleftrightarrow{AD}$ and $\overleftrightarrow{BC}$. Let $E$ be the point on $\overleftrightarrow{AB}$ such that $\overleftrightarrow{PE}$ and $\overleftrightarrow{CD}$ are parallel, and let $F$ be the point on $\overleftrightarrow{CD}$ such that $\overleftrightarrow{PF}$ and $\overleftrightarrow{AB}$ are parallel. Suppose that a point $Y$ on $CD$ and a point $Z$ on $AB$ are such that $P$, $Y$ and $Z$ are collinear. Then $\frac{|EA|}{|EB|}=\frac{|FC|}{|FD|}$, and letting $\alpha^2=|EA||EB|$ and $\beta^2=|FC||FD|$ we have $\frac{|EZ|}{\alpha}\cdot\frac{|FY|}{\beta}=1$.
\end{lemma}

\begin{proof}
The triangles $EPX$ and $FYP$ are similar, so $\frac{|EZ|}{|EP|}=\frac{|FP|}{|FY|}$, so $|EZ||FY|=|FP||EP|$. In particular, when $Z=A$, we have $|EA||FD|=|FP||EP|$, and when $Z=B$, we have $|EB||FC|=|FP||EP|$, so $|EA||FD|=|EB||FC|$ and $\frac{|EA|}{|EB|}=\frac{|FC|}{|FD|}$. 

Now $|EZ||FY|=|FP||EP|=|EA||FD|=|EB||FC|$, so:

\[|EZ||FY|=\sqrt{|EA||FD||EB||FC|}=\sqrt{\alpha^2\beta^2}=\alpha\beta\]
\end{proof}

The formulas $\frac{|EZ|}{\alpha}$ and $\frac{|EY|}{\beta}$ give mappings from points on the input segments to $\rr$. We want the input segment to correspond to the intervals $[\frac12, 2]$, so we choose $|EB|=4|EA|$ (and therefore $|FD|=4|FC|$), so $\alpha=2|EA|=\frac12|EB|$ and $\beta=2|FC|=\frac12|FD|$. This means that $\frac{|EZ|}{\alpha}$ and $\frac{|FY|}{\beta}$ will map the segments $AB$ and $CD$ respectively onto $[\frac12, 2]$.

On its own, \Cref{lem:Inversion} could be used to construct a nook which enforces an inversion constraint on two input segments which aren't parallel. The input segments to a constraint gadget should be parallel, so we will need the result of \Cref{lem:AngleCopy} to correct the orientation.

\begin{figure}
\centering
\includegraphics[page=21]{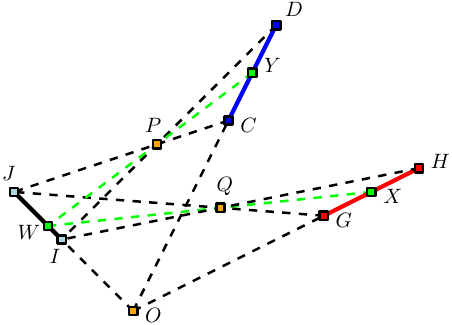}
\caption{If $X$ is on $GH$, then define $Y$ by projecting $X$ onto $IJ$ through $Q$ and then projecting that point onto $CD$ through $P$. The relative position of $X$ on $GH$ is the same as the relative position of $Y$ on $CD$.}
\label{fig:InversionLemma2}
\end{figure}

\begin{lemma}\label{lem:AngleCopy}
Suppose line segments $GH$, $CD$, and $IJ$ are such that $\overleftrightarrow{GH}$, $\overleftrightarrow{CD}$, and $\overleftrightarrow{IJ}$ all intersect at a point $O$, as in \Cref{fig:InversionLemma2}. Also suppose that the ratios $\frac{|OG|}{|OH|}$ and $\frac{|OC|}{|OD|}$ are the same. Let $P$ be the point where $ID$ and $JC$ intersect, and $Q$ be the point where $IH$ and $JG$ intersect. Suppose that points $X$ on $GH$, $Y$ on $CD$ and $W$ on $IJ$ are such that $W$, $P$, and $Y$ are collinear and $W$, $Q$ and $X$ are collinear. Then $\frac{|GX|}{|GH|}=\frac{|CY|}{|CD|}$.
\end{lemma}

\begin{proof}
Place the figure in the vector space $\rr^2$ with the point $O$ at $(0, 0)$. We will show that there is a linear map which sends points $G, X$ and $H$ to $C, Y$ and $D$ respectively.

The pairs of vectors $\{I, G\}$ and $\{I, C\}$ are each bases for $\rr^2$. Let $\theta$ be the linear isomorphism $\rr^2\rightarrow\rr^2$ which sends a vector $V$ to the coefficients $(t, s)$ such $V=tI+sG$, and let $\psi$ be a similar map which writes $V$ as $tI+sC$. Now the linear map $\psi^{-1}\circ \theta$ fixes points on the line containing $I$ and $J$, and sends $G$ to $C$. Since $\frac{|OA|}{|OB|}=\frac{|OC|}{|OD|}$, this map sends $H$ to $D$, and so sends $Q$ to $P$. The point $W$ is fixed, so the line $\overleftrightarrow{WQ}$ is sent to $\overleftrightarrow{WP}$, meaning that $X$ is sent to $Y$. So points $G, X$ and $H$ are sent to $C, Y$ and $D$ respectively. Since linear maps preserve ratios of distances along a line, we see that:
\[\frac{|GX|}{|GH|}=\frac{|CY|}{|CD|}\]
\end{proof}

Note that the mappings from $GH$ to $IJ$ and from $IJ$ to $CD$ are in general \emph{not} linear, only the composition is. 

\begin{lemma}\label{lem:InversionGadgets}
In either the $xy\ge 1$ or $x\left(\frac{5}2-y\right)\le 1$ gadget, suppose there are points $C$ on $\overleftrightarrow{JP}$ and $D$ on $\overleftrightarrow{IP}$ such that $\overleftrightarrow{GH}$, $\overleftrightarrow{CD}$ and $\overleftrightarrow{IJ}$ all intersect at a point $O$, with $C$ and $D$ satisfying:

\[\frac{|OG|}{|OH|}=\frac{|OC|}{|OD|}\]

Suppose also that there are points $E$ on $\overleftrightarrow{AB}$ and $F$ on $\overleftrightarrow{CD}$ so that $A$ is between $E$ and $B$, $C$ is between $F$ and $D$, $\overleftrightarrow{PF}$ and $\overleftrightarrow{AB}$ are parallel and $\overleftrightarrow{PE}$ and $\overleftrightarrow{CD}$ are parallel, with $|FD|=4|FC|$ and $|EB|=4|EA|$. Then:

\[\left(\frac32\cdot\frac{|AW|}{|AB|}+\frac12\right)\left(\frac32\cdot\frac{|GX|}{|GH|}+\frac12\right)=1\]
so the two inversion gadgets enforce the appropriate constraints.
\end{lemma}

\begin{proof}
Examples of points $C$, $D$, $E$ and $F$ with these properties are shown in \Cref{fig:InversionGESkeleton,fig:InversionLESkeleton}. Let $Z$ be the intersection of $\overleftrightarrow{YP}$ with $CD$. By \Cref{lem:AngleCopy}:

\[\frac{|GX|}{|GH|}=\frac{|CZ|}{CD}\]

By \Cref{lem:Inversion}, we have that:

\[\frac{|EW|}{2|EA|}\cdot\frac{|FZ|}{2|FC|}=1\]

Since $4|EA|=|EB|=\frac43|AB|$, we have that:

\[\left(\frac32\cdot\frac{|AW|}{|AB|}+\frac12\right)=\left(\frac{|AW|}{2|EA|}+\frac{|EA|}{2|EA|}\right)=\frac{|EW|}{2|EA|}\]

Similarly:

\[\left(\frac32\cdot\frac{|GX|}{|GH|}+\frac12\right)=\left(\frac32\cdot\frac{|CZ|}{|CD|}+\frac12\right)=\frac{|FZ|}{2|FC|}\]

So:

\[\left(\frac32\cdot\frac{|AW|}{|AB|}+\frac12\right)\left(\frac32\cdot\frac{|GX|}{|GH|}+\frac12\right)=1\]
\end{proof}

It remains to show that gadgets satisfying the conditions of \Cref{lem:InversionGadgets} can actually be created. The most direct approach would require solving a quadratic equation, potentially introducing irrational coordinates. The authors of \cite{ExistsRHardness} solved a similar issue with their inversion gadgets by carefully choosing coordinates that yield quadratic equations with rational solutions. However, our gadgets can be constructed more geometrically.

\begin{lemma}\label{lem:InversionConstruction}
Geometry for the $xy\ge 1$ and $x\left(\frac{5}2-y\right)\le 1$ gadgets can be constructed in a way that satisfies the conditions of \Cref{lem:InversionGadgets,lem:Region1,lem:Region2}. This does not require irrational coordinates.
\end{lemma}

\begin{proof}
First we show how to construct the $xy\ge 1$ gadget.

Start with the two input segments $AB$ and $GH$. Place the segment $CD$ so that it intersects $AB$ one-third of the way up (at the point representing $1$), and so that $C$ is on $\overleftrightarrow{AG}$ and $D$ is $\overleftrightarrow{BH}$. In \Cref{fig:InversionGESkeleton}, the slope of $CD$ is $-1$. Letting $O$ be the intersection of $\overleftrightarrow{GH}$ and $\overleftrightarrow{CD}$, it is straightforward to check that $\frac{|OG|}{|OH|}=\frac{|OC|}{|OD|}$. 

Now let $P$ be the intersection of $\overleftrightarrow{BC}$ and $\overleftrightarrow{AD}$. Let $E$ on $\overleftrightarrow{AB}$ and $F$ on $\overleftrightarrow{CD}$ be the points such that $\overleftrightarrow{PF}$ and $\overleftrightarrow{AB}$ are parallel and $\overleftrightarrow{PE}$ and $\overleftrightarrow{CD}$ are parallel. 

There is a line through $P$ and the intersection point $Z$ of $AB$ and $CD$. The intersection point is $\frac13$ of the way along $AB$ and is the same fraction of the way along $CD$. So this point must correspond to $1$ under the mapping from \Cref{lem:Inversion}, that is $\frac{EZ}{\alpha}=\frac{FZ}{\beta}=1$. This means that $|EB|=4|EA|$ and $|FD|=4|FC|$, as required.

Now the nook can be created by letting $I$ and $J$ be the intersections of $\overleftrightarrow{AP}$ and $\overleftrightarrow{BP}$ respectively with the horizontal line through $O$. The point $Q$ is the intersection of $\overleftrightarrow{IH}$ and $\overleftrightarrow{JG}$. 

The left wall of the constraint gadget can be placed far enough left as to not intersect $R_1$. Since $DI$ has slope steeper than $DC$ and $DC$ has slope $-1$, the nook segment does not intersect the region $R_2$. So the gadget created satisfies the conditions of \Cref{lem:Region1,lem:Region2}.

Next, we construct the $x\left(\frac{5}2-y\right)\le 1$ gadget. First, choose the position of the collinear points $C$, $D$, and $F$ so that $C$ is between $F$ and $D$ and $|FD|=4|FC|$. In \Cref{fig:LEStep1}, these points are chosen to lie on a line with slope $-1$. Now let $P$ be a point on the vertical line through $F$ (for example, the point on this line that is also on the horizontal line through $D$). Now choose a point $O$ on $\overleftrightarrow{CD}$ such that $D$ is between $O$ and $C$, and choose points $G$ and $H$ on the vertical line through $O$ so that $\frac{|OH|}{|OG|}=\frac{|OD|}{|OC|}$. These should be chosen to be below the line $\overleftrightarrow{PD}$. This is shown in \Cref{fig:LEStep1}.

\begin{figure}
\centering
\includegraphics[page=22]{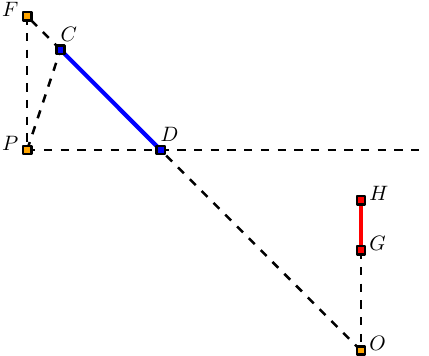}
\caption{If we choose $G$ close enough to $O$, then the position of $H$ determined by $\frac{|OH|}{|OG|}=\frac{|OD|}{|OC|}$ will be below $\overleftrightarrow{PD}$.}
\label{fig:LEStep1}
\end{figure}

Next, let $AB$ be the vertical line segment with $A$ on $\overleftrightarrow{PD}$ and $B$ on $\overleftrightarrow{PC}$ with $|AB|=|GH|$ and let $E$ be the point below $A$ on $\overleftrightarrow{AB}$ such that $\overleftrightarrow{PE}$ is parallel to $\overleftrightarrow{CD}$. By \Cref{lem:Inversion}, $\frac{|EA|}{|EB|}=\frac{|FC|}{|FD|}=\frac{1}{4}$. This is shown in \Cref{fig:LEStep2}.

\begin{figure}
\centering
\includegraphics[page=23]{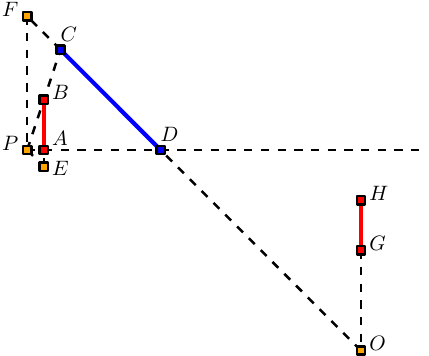}
\caption{Affine transformations scale parallel segments by the same amount. We want to obtain the final gadget by an affine transform of this geometry, so we should make sure that $|AB|=|GH|$. As long as this condition is satisfied, it is possible to find an appropriate affine transform.}
\label{fig:LEStep2}
\end{figure}

Finally, perform an affine transformation mapping $AB$ and $GH$ to the positions of the input segments of the gadget. Since affine transformations preserve ratios of lengths along lines, the new geometry still satisfies the conditions of \Cref{lem:InversionGadgets}. Since the segment $GH$ is below $\overleftrightarrow{PD}$, the nook does not obstruct itself.

The remaining points $I$, $J$, and $Q$ can be placed as in the construction of the $xy\ge 1$ gadget. It is straightforward to check that the gadget satisfies the conditions of \Cref{lem:Region1,lem:Region2}.

Nothing here requires irrational coordinates.
\end{proof}

\subsection{Addition gadgets}\label{sec:AdditionVerification}

Each nook creates a constraint between the positions of only $2$ guards, but the addition gadgets should create a constraint involving $3$ variables. This can be accomplished by allowing both coordinates of one guard to vary, as illustrated in \Cref{fig:TripleConstraint}.

\begin{figure}
\centering
\includegraphics[page=24]{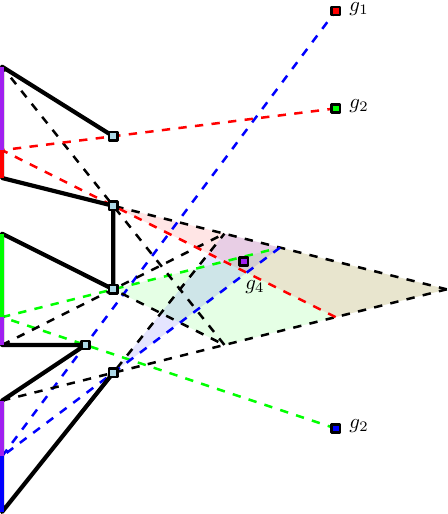}
\caption{Creating a constraint between $3$ variables. The guard $g_1$, $g_2$ and $g_3$ each see part of one of the nook segments. If the three shaded regions intersect, then a guard $g_4$ can be placed in the intersection so that $g_4$ guards the parts of the nook segments that aren't guarded by $g_1$, $g_2$ and $g_3$. If the three shaded regions do not intersect, then this isn't possible.}
\label{fig:TripleConstraint}
\end{figure}

The addition gadgets are shown in \Cref{fig:AdditionGE,fig:AdditionLE}. Each gadget has $3$ input segments, an auxiliary guard region, and three constraint nooks. Each constraint nook creates a constraint between the auxiliary guard and one input segment. 

\begin{figure}
\centering
\includegraphics[page=25]{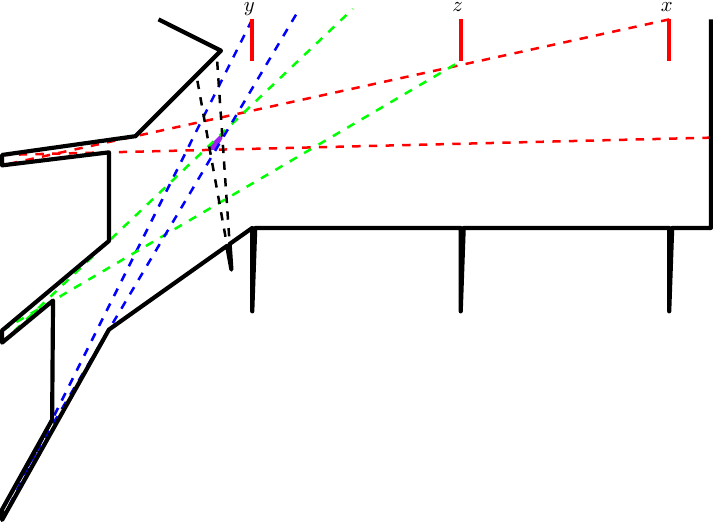}
\caption{The $x+y\le z$ addition gadget. The partial visibility regions of the $3$ nooks are marked, as is the visibility region of the wedge forming the auxiliary guard region. The auxiliary guard will need to be somewhere in the purple shaded region.}
\label{fig:AdditionGE}
\end{figure}

\begin{figure}
\centering
\includegraphics[page=26]{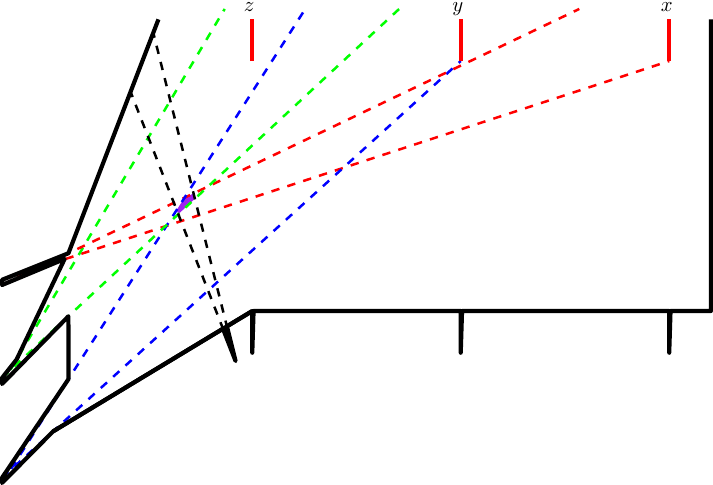}
\caption{The $x+y\ge z$ addition gadget.}
\label{fig:AdditionLE}
\end{figure}

The gadget can be guarded when it is possible to place a guard in the auxiliary guard region so that it can see the parts of the nook segments that aren't visible to the input guards. We can use \Cref{lem:CopyLemma} to control the relationship between the positions of the guards on the input segments and their ``shadows'' on the nook segments. \Cref{lem:Addition} allows us to determine what parts of these nook segments can be seen by an auxiliary guard. 

\begin{lemma}\label{lem:Addition}
Suppose the line segments $AB$, $CD$, and $EF$ have the same length and lie on the same vertical line, as in \Cref{fig:SkewAddition}, and suppose $|CB|=|DE|$. Let points $P$, $Q$, and $R$ lie on a vertical line, with $|QP|=|QR|$. Note that $\overleftrightarrow{AP}$, $\overleftrightarrow{CQ}$, and $\overleftrightarrow{ER}$ intersect in a single point, and the same is true of $\overleftrightarrow{BP}$, $\overleftrightarrow{DQ}$, and $\overleftrightarrow{FR}$.

Suppose points $X$ and $Y$ lie on $AB$ and $EF$ respectively. Now $\overleftrightarrow{XP}$ and $\overleftrightarrow{YQ}$ intersect at a point $I$. Let $Z$ be the point on the intersection of $CD$ and $\overleftrightarrow{IQ}$. 

If all is as above, then $\frac12\left(|AX|+|EY|\right)=|CZ|$.
\end{lemma}

\begin{figure}
\centering
\includegraphics[page=27]{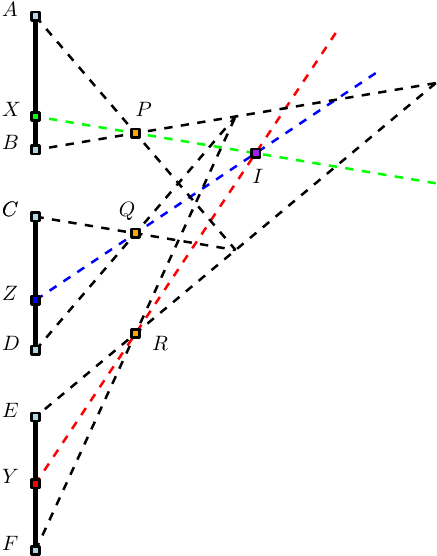}
\caption{By \Cref{lem:Addition}, $\frac12\left(|AX|+|EY|\right)=|CZ|$.}
\label{fig:SkewAddition}
\end{figure}

\begin{proof}
The idea is to transform the geometry from \Cref{fig:SkewAddition} to obtain something like \Cref{fig:TransformedAddition}.

\begin{figure}
\centering
\includegraphics[page=28]{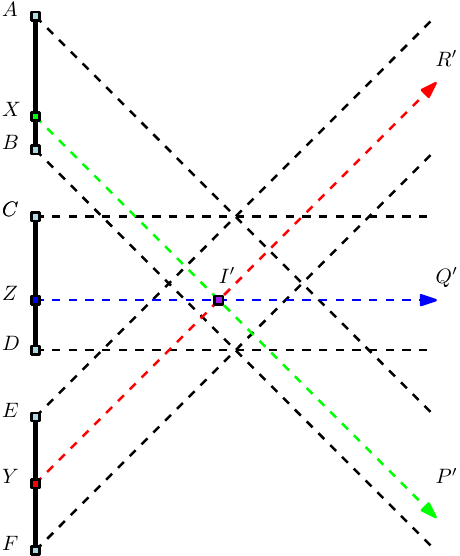}
\caption{A projective transformation sends points $P$, $Q$, and $R$ to points $P'$, $Q'$ and $R'$ at infinity, so the family of lines through any one of these points is sent to a family of parallel lines.}
\label{fig:TransformedAddition}
\end{figure}

This transformation should send straight lines to straight lines and should fix the line containing points $A, B, C, D, E, F, X, Y$ and $Z$. Additionally, the points $P$, $Q$ and $R$ should be sent to ``infinity'', lines through $P$ should be sent to lines with slope $-1$, lines through $Q$ should be sent to lines with slope $0$, and lines through $R$ should be sent to lines with slope $+1$. If such a transformation exists, then it is clear by straightforward linear algebra that $\frac12\left(|AX|+|EY|\right)=|CZ|$.

Isomorphisms of the projective space $\rr P^2$ send straight lines to straight lines. A degrees-of-freedom argument would be sufficient to find a transformation of projective space with the required properties. Instead we just give it explicitly. Assume that the line containing $A$ and $B$ is vertical, and let $x_0$ be the $x$-coordinate of this line. Let $x_1$ be the $x$-coordinate of $P, Q$ and $R$, $y_0$ the $y$-coordinate of $Q$, and let $a=|QP|=|QR|$, so $P$ has $y$-coordinate $y_0+a$ and $R$ has $y$-coordinate $y_0-a$. Then the transformation $(x, y)\rightarrow (x', y')$ with the desired properties is defined by:

\[\begin{bmatrix}x_0+a&0&-(x_1+a)x_0\\y_0&x_0-x_1&-y_0x_0\\1&0&-x_1\end{bmatrix}\begin{bmatrix}x\\y\\1\end{bmatrix}=\lambda\begin{bmatrix}x'\\y'\\1\end{bmatrix}\]


Writing the map in this form makes it easy to check what happens to lines through $P$, $Q$, and $R$. In particular, for a $3\times3$ matrix $A$, if:

\[A\begin{bmatrix}p_x\\p_y\\1\end{bmatrix}=\begin{bmatrix}a\\b\\0\end{bmatrix}\]

\noindent then the map $(x, y)\rightarrow (x', y')$ given by:

\[\lambda\begin{bmatrix}x'\\y'\\1\end{bmatrix}=A\begin{bmatrix}x\\y\\1\end{bmatrix}\]

\noindent sends lines through $(p_x, p_y)$ to lines parallel to $\begin{bmatrix}a\\b\end{bmatrix}$. 

Note that the setup used here is very similar the addition gadgets in \cite{ExistsRHardness}, although our verification is different.
\end{proof}

We want to create gadgets enforcing $x+y=z$, not $\frac12(x+y)=z$, so we need to use \Cref{lem:CopyLemma} to change the scale and offset of $z$. A schematic of the nooks used by the addition gadgets are shown in \Cref{fig:AdditionNooks,fig:OtherAdditionNooks}.

\begin{figure}
\centering
\includegraphics[page=29]{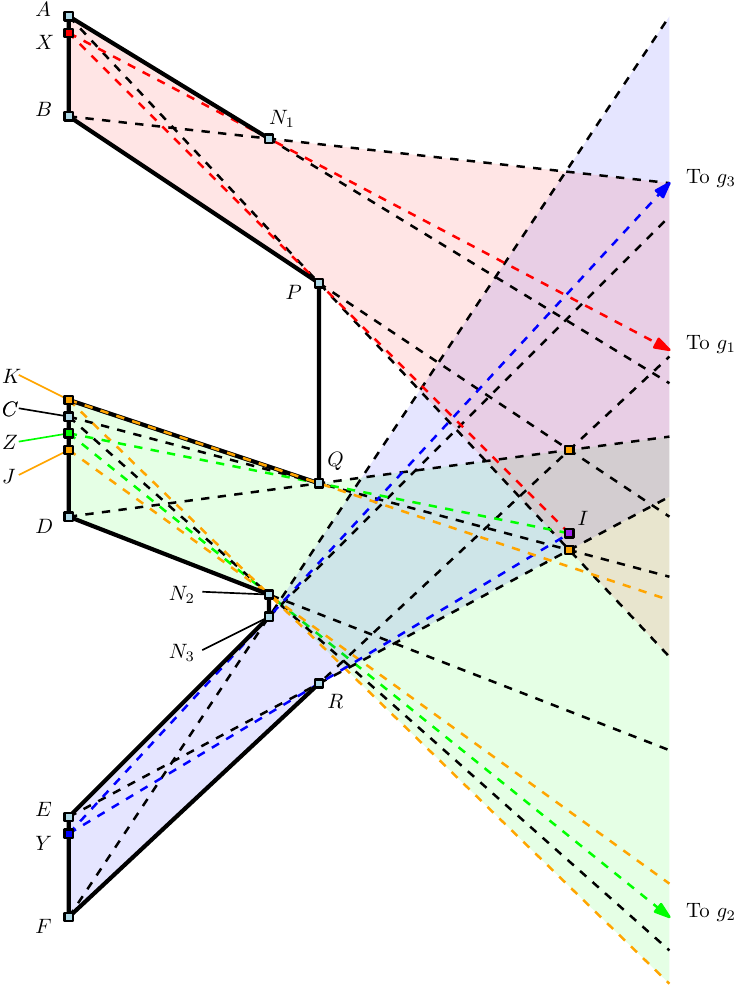}
\caption{Schematic of the $x+y\le z$ addition gadget. In the actual gadgets, the lines bounding all the visibility regions have positive slopes as shown in \Cref{fig:AdditionLE,fig:AdditionGE}, but it is not possible to draw this faithfully at a readable scale. }
\label{fig:AdditionNooks}
\end{figure}

\begin{figure}
\centering
\includegraphics[page=30]{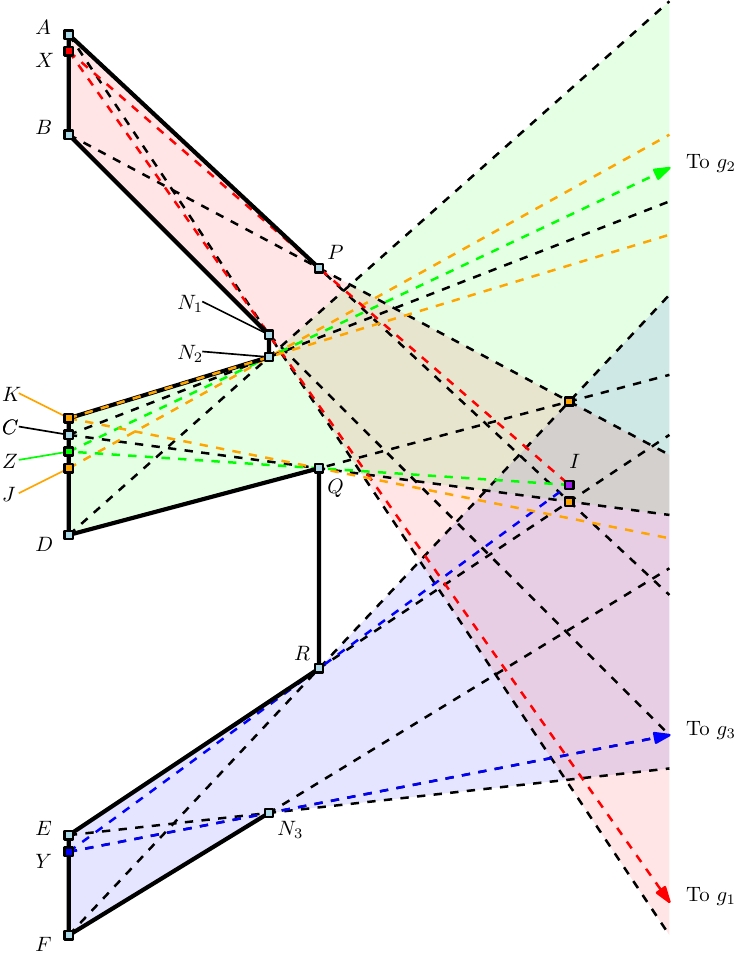}
\caption{The nooks for the $x+y\ge 1$ gadget.}
\label{fig:OtherAdditionNooks}
\end{figure}

\begin{lemma}\label{lem:AdditionVerification}
It is possible to construct addition gadgets enforcing constraints $x+y\ge z$ or $x+y\le z$.
\end{lemma} 

\begin{proof}
We construct the $x+y\le z$ addition gadget. The $x+y\ge z$ gadget is essentially identical (see \Cref{fig:OtherAdditionNooks}).

We should place $A$, $B$, $C$, $D$, $E$, and $F$ are all on the same vertical line with $|AB|=|CD|=|EF|$ and $|BC|=|DE|$. Let $K$ and $J$ be the points on this line so that $(K, C, J, D)$ can be linearly mapped onto $\left(\frac12, 1, 2, 4\right)$, that is $C$ is between $K$ and $D$ with $6|CD|=|KD|$ and $J$ is between $C$ and $D$ with $3|CJ|=|CD|$.

Let $P, Q$ and $R$ lie on a vertical line with $|PQ|=|QR|$, and create nooks with opening points $P$, $Q$ and $R$ and with nook segments $AB$, $KD$, and $EF$ respectively, so that $P$ restricts $B$, $Q$ restricts $C$, and $R$ restricts $F$. Let $N_1$, $N_2$, and $N_3$ be the other opening points of these nooks.

The guards $g_1, g_2, g_3$ on the vertical input segments represent values $x, z, y\in [\frac12, 2]$ respectively. The guard segment for $x$ should have endpoints on $\overleftrightarrow{AN_1}$ and $\overleftrightarrow{BN_1}$, the guard segment for $y$ should have endpoints on $\overleftrightarrow{EN_3}$ and $\overleftrightarrow{FN_3}$, and the guard segment for $z$ should have endpoints on $\overleftrightarrow{KN_2}$ and $\overleftrightarrow{JN_2}$.

We should make sure that the $3$ pairs of lines $\left(\overleftrightarrow{AN_1}, \overleftrightarrow{BP}\right)$, $\left(\overleftrightarrow{LN_2}, \overleftrightarrow{KQ}\right)$ and $\left(\overleftrightarrow{EN_3}, \overleftrightarrow{FR}\right)$ are parallel or intersect to the left of the vertical line containing the nook segments. This ensures that the guard segment for $g_1$ is above $\overleftrightarrow{BP}$, the guard segment for $g_2$ is above $\overleftrightarrow{FR}$, and the guard segment for $g_3$ is below $\overleftrightarrow{CQ}$.

Constructing the addition gadgets to these specifications just requires making the nooks small enough and far away enough that the walls of each nook don't obstruct any of the other nooks, and so that the visibility region of each nook only intersects one of the input segments. It is also straightforward to ensure that the conditions of \Cref{lem:Region1,lem:Region2} are satisfied.

Now we show that, if all is as above, and each guard segment only intersects the partial visibility region of one nook segment, then it is possible to guard the nooks with a single additional guard if and only $x+y\le z$.

Let $X$ be the intersection of $AB$ and the line through $g_1$ and $N_1$. Let $Y$ by the intersection of $EF$ and the line through $g_2$ and $N_3$. Since the $g_1$ is above $\overleftrightarrow{BP}$ and $g_2$ is above $\overleftrightarrow{FR}$, $g_1$ can see all the points below $X$ on $AB$ and $g_2$ can see all the points below $Y$ on $EF$. By \Cref{lem:CopyLemma}:

\[\frac32\cdot\frac{|AX|}{|AB|}+\frac12=x\text{ and }\frac32\cdot\frac{|EY|}{|EF|}+\frac{12}=y\]

One guard can see the remaining parts of $AB$ and $EF$, if and only if it is on or above the lines $\overleftrightarrow{XP}$ and $\overleftrightarrow{YR}$. Let $I$ be the intersection point of $\overleftrightarrow{XP}$ and $\overleftrightarrow{YR}$ and let $Z$ be the intersection of $\overleftrightarrow{CJ}$ and $\overleftrightarrow{IQ}$.

The auxiliary guard can see the points on $CJ$ below $Z$ (if it is placed at $I$), but can't ever see points on $CJ$ above $Z$ while still guarding the upper parts of $AB$ and $EF$.

By \Cref{lem:Addition}, $|AX|+|EY|=2|CZ|$. Since $|AB|=|EF|=|CD|$:

\[x+y=3\frac{|AX|+|EY|}{|CD|}+1=6\frac{CZ}{|CD|}+1=\frac{3}{2}\cdot\frac{|KZ|}{|CJ|}+\frac{1}{2}\]

Since $g_3$ is always below $\overleftrightarrow{CQ}$, the guard $g_3$ can see the rest of $CD$ if and only if it is on or above the line $\overleftrightarrow{ZN_2}$. By \Cref{lem:CopyLemma}, this is when:

\[z\ge \frac{3}{2}\cdot\frac{|KZ|}{|CJ|}+\frac{1}{2}\]

So the nooks can be guarded if and only if $x+y\le z$. 
\end{proof}

\subsection{Final verification}

\begin{lemma}
Given an instance $\Phi$ of ETR-INV-REV, it is possible to construct in polynomial time an art gallery $P$ and a number $n$ so that $P$ can be guarded by $n$ guards if and only if $\Phi$ is satisfiable.
\end{lemma}

\begin{proof}
First, create corresponding constraint gadgets for every constraint of form $xy\le 1$, $x\left(\frac{5}2-y\right)\le 1$, $x+y\le z$, or $x+y\ge z$ that occurs in $\Phi$. Arrange these in a horizontal row along the bottom of the art gallery so that all the input segments occur in a horizontal row and the distance between any pair of input segments is a multiple of $d=5$ times the length of a single input segment.

The construction in \Cref{sec:CopyNooks} can be used to place the copy nooks and variable segments. Now connect all the pieces of the art gallery and add the wedges forming all the guard segments. These wedges should be chosen to be sufficiently narrow so that visibility regions of wedges from different guard regions do not intersect. 

It is straightforward to check that this construction can be done in polynomial time and that the coordinates needed are rational numbers with at most polynomially many bits.

There are $n$ guard regions, consisting of the variable segments, input segments, and auxiliary guards. By \Cref{lem:GuardRegions}, any guarding configuration with $n$ guards must have one guard in each guard region. Conversely, any guarding configuration with one guard in each guard can clearly guard all of these wedges.

There are some walls of the art gallery that aren't part of the wedges, copy nooks, or constraint gadgets. The top, bottom, and right walls of the art gallery can be guarded by any guard on a variable segment. For any two copy nooks, a segment connecting can always be guarded by either of the guards on the variable segments in the visibility regions of those nooks. So any guarding configuration with one guard in each guard region will guard everything expect for possibly the copy nook and constraint gadgets.

By \Cref{lem:InversionGadgets,lem:AdditionVerification}, each constraint gadget can be guarded if and only if the guards on its input segments represent variables satisfying the corresponding constraint. Whenever a variable in $\Phi$ appears in a constraint, there are copy nooks between that variable segment and an input segment to the gadget representing that constraint. By \Cref{lem:CopyNook}, these copy nooks are guarded if and only if the guard on the input segment and the guard on the variable segment represent the value in $[\frac12, 2]$.

So $P$ can be guarded by $n$ guards if and only if $\Phi$ has a satisfying assignment. 
\end{proof}

The $\exists\rr$-hardness of the point-boundary art-gallery problem follows from the $\exists\rr$-hardness of ETR-INV-REV, proving \Cref{BGMain}.

\section{Conclusions}\label{Conclusion}

Our result shows that the X-Y and Y-X variants of the art gallery problem are equivalent under polynomial time reductions when $X, Y\in \{\text{Vertex}, \text{Point}, \text{Boundary}\}$. Visibility is reflexive, so this result seems unsurprising. Maybe it could be proven more directly.

It is interesting to note that our construction, while intended for the point-boundary variant of the art gallery problem, is also sufficient to show the $\exists\rr$-hardness of the standard point-point variant. When it is possible to guard an art gallery produced by our construction in the point-boundary setting, the same guarding configuration is also a point-point guarding configuration. Something similar happens in \cite{ExistsRHardness}, where the construction for the boundary-point variant is also a construction for the point-point variant. It doesn't seem to be possible to adapt either construction to the boundary-boundary variant.

Unlike the construction from \cite{ExistsRHardness}, the nook segments in our construction can be chosen to all be axis-parallel. Each other segment is always entirely guarded by a single guard. It may be possible to adapt our construction to show that guarding orthogonal polygons is $\exists\rr$-hard.

\bibliographystyle{plainurl}
\bibliography{GalleryVariants}

\begin{thebibliography}{10}

\bibitem{IrrationalGuards}
Mikkel Abrahamsen, Anna Adamaszek, and Tillmann Miltzow.
\newblock {Irrational Guards are Sometimes Needed}.
\newblock In Boris Aronov and Matthew~J. Katz, editors, {\em 33rd International Symposium on Computational Geometry (SoCG 2017)}, volume~77 of {\em Leibniz International Proceedings in Informatics (LIPIcs)}, pages 3:1--3:15, Dagstuhl, Germany, 2017. Schloss Dagstuhl -- Leibniz-Zentrum f{\"u}r Informatik.
\newblock URL: \url{https://drops.dagstuhl.de/entities/document/10.4230/LIPIcs.SoCG.2017.3}, \href {https://doi.org/10.4230/LIPIcs.SoCG.2017.3} {\path{doi:10.4230/LIPIcs.SoCG.2017.3}}.

\bibitem{ExistsRHardness}
Mikkel Abrahamsen, Anna Adamaszek, and Tillmann Miltzow.
\newblock The art gallery problem is $\exists\rr$-complete.
\newblock {\em Journal of the ACM}, 69(1):1–70, 2021.
\newblock \href {https://doi.org/10.1145/3486220} {\path{doi:10.1145/3486220}}.

\bibitem{PackingProblem}
Mikkel Abrahamsen, Tillmann Miltzow, and Nadja Seiferth.
\newblock Framework for er-completeness of two-dimensional packing problems.
\newblock In {\em 2020 IEEE 61st Annual Symposium on Foundations of Computer Science (FOCS)}, pages 1014--1021, 2020.
\newblock \href {https://doi.org/10.1109/FOCS46700.2020.00098} {\path{doi:10.1109/FOCS46700.2020.00098}}.

\bibitem{AlmostConvex}
Akanksha Agrawal, Kristine V.~K. Knudsen, Daniel Lokshtanov, Saket Saurabh, and Meirav Zehavi.
\newblock {The Parameterized Complexity of Guarding Almost Convex Polygons}.
\newblock In Sergio Cabello and Danny~Z. Chen, editors, {\em 36th International Symposium on Computational Geometry (SoCG 2020)}, volume 164 of {\em Leibniz International Proceedings in Informatics (LIPIcs)}, pages 3:1--3:16, Dagstuhl, Germany, 2020. Schloss Dagstuhl--Leibniz-Zentrum f{\"u}r Informatik.
\newblock URL: \url{https://drops.dagstuhl.de/opus/volltexte/2020/12161}, \href {https://doi.org/10.4230/LIPIcs.SoCG.2020.3} {\path{doi:10.4230/LIPIcs.SoCG.2020.3}}.

\bibitem{HomotopyUniversality}
Daniel Bertschinger, Nicolas~El Maalouly, Tillmann Miltzow, Patrick Schnider, and Simon Weber.
\newblock Topological art in simple galleries.
\newblock In {\em Symposium on Simplicity in Algorithms ({SOSA})}, pages 87--116. Society for Industrial and Applied Mathematics, jan 2022.
\newblock URL: \url{https://doi.org/10.1137%2F1.9781611977066.8}, \href {https://doi.org/10.1137/1.9781611977066.8} {\path{doi:10.1137/1.9781611977066.8}}.

\bibitem{SubsetPSPACE}
John Canny.
\newblock Some algebraic and geometric computations in pspace.
\newblock In {\em Proceedings of the Twentieth Annual ACM Symposium on Theory of Computing}, STOC '88, page 460–467, New York, NY, USA, 1988. Association for Computing Machinery.
\newblock \href {https://doi.org/10.1145/62212.62257} {\path{doi:10.1145/62212.62257}}.

\bibitem{FirstExact}
Alon Efrat and Sariel Har-Peled.
\newblock Guarding galleries and terrains.
\newblock {\em Inf. Process. Lett.}, 100:238--245, 01 2006.
\newblock \href {https://doi.org/10.1007/978-0-387-35608-2_16} {\path{doi:10.1007/978-0-387-35608-2_16}}.

\bibitem{BoundaryNPHard}
Aldo Laurentini.
\newblock Guarding the walls of an art gallery.
\newblock {\em The Visual Computer}, 15:265--278, 10 1999.
\newblock \href {https://doi.org/10.1007/s003710050177} {\path{doi:10.1007/s003710050177}}.

\bibitem{NPHardness}
D.~T. Lee and Arthur~K. Lin.
\newblock Computational complexity of art gallery problems.
\newblock {\em {IEEE} Transactions on Information Theory}, 32(2):276--282, 1986.
\newblock \href {https://doi.org/10.1109/TIT.1986.1057165} {\path{doi:10.1109/TIT.1986.1057165}}.

\bibitem{ConvexConcave}
Tillmann Miltzow and Reinier~F. Schmiermann.
\newblock On classifying continuous constraint satisfaction problems.
\newblock In {\em 2021 IEEE 62nd Annual Symposium on Foundations of Computer Science (FOCS)}, pages 781--791, 2022.
\newblock \href {https://doi.org/10.1109/FOCS52979.2021.00081} {\path{doi:10.1109/FOCS52979.2021.00081}}.

\bibitem{Mnev}
Nikolai~E Mn{\"e}v.
\newblock The universality theorems on the classification problem of configuration varieties and convex polytopes varieties.
\newblock In {\em Topology and geometry—Rohlin seminar}, pages 527--543. Springer, 1988.

\bibitem{ArtGalleryTextbook}
Joseph O'Rourke.
\newblock {\em Art gallery theorems and algorithms}, volume~57.
\newblock Oxford New York, NY, USA, 1987.

\bibitem{ETRRobustness}
Marcus Schaefer and Daniel Stefankovic.
\newblock Fixed points, nash equilibria, and the existential theory of the reals.
\newblock {\em Theory of Computing Systems}, 60:172--193, 2017.

\bibitem{RealHeirarchy}
Marcus Schaefer and Daniel \v{S}tefankovi\v{c}.
\newblock Beyond the existential theory of the reals.
\newblock {\em Theor. Comp. Sys.}, 68(2):195–226, December 2023.
\newblock \href {https://doi.org/10.1007/s00224-023-10151-x} {\path{doi:10.1007/s00224-023-10151-x}}.

\bibitem{Stretchability}
Peter~W. Shor.
\newblock Stretchability of pseudolines is np-hard.
\newblock In {\em Applied Geometry And Discrete Mathematics}, 1990.

\bibitem{TopologyUniversality}
Jack Stade and Jamie Tucker-Foltz.
\newblock {Topological Universality of the Art Gallery Problem}.
\newblock In Erin~W. Chambers and Joachim Gudmundsson, editors, {\em 39th International Symposium on Computational Geometry (SoCG 2023)}, volume 258 of {\em Leibniz International Proceedings in Informatics (LIPIcs)}, pages 58:1--58:13, Dagstuhl, Germany, 2023. Schloss Dagstuhl -- Leibniz-Zentrum f{\"u}r Informatik.
\newblock URL: \url{https://drops.dagstuhl.de/entities/document/10.4230/LIPIcs.SoCG.2023.58}, \href {https://doi.org/10.4230/LIPIcs.SoCG.2023.58} {\path{doi:10.4230/LIPIcs.SoCG.2023.58}}.

\end{thebibliography}

\end{document}